\documentclass{article}
\usepackage{amsfonts,amsmath,amssymb,amsthm,fullpage}
\usepackage[pdftex]{graphicx}
\usepackage[ruled]{algorithm2e}

\usepackage{algorithmic}
\theoremstyle{plain}
\newtheorem{theorem}{Theorem}
\newtheorem{lemma}{Lemma}
\theoremstyle{definition}
\newtheorem{definition}{Definition}
\newtheorem{property}{Property}
\newcommand{\zero}[1]{\makebox[0mm][l]{$#1$}}

\newcommand{\istrut}[2][0]{\rule[- #1 mm]{0mm}{#1 mm}\rule{0mm}{#2 mm}}
\newcommand{\rb}[2]{\raisebox{#1 mm}[0mm][0mm]{#2}}

\newcommand{\Egervary}{Egerv\'{a}ry}

\newcommand{\paren}[1]{\left( #1 \right)}

\newcommand{\ceil}[1]{\lceil #1 \rceil}

\newcommand{\floor}[1]{\lfloor #1 \rfloor}
\newcommand{\f}[2]{\frac{#1}{#2}}
\newcommand{\fr}[2]{\mbox{$\frac{#1}{#2}$}}

\newcommand{\bydef}{\stackrel{\rm def}{=}}
\newcommand{\poly}{\mathrm{poly}}
\DeclareMathOperator*{\argmin}{arg\,min}

\newcommand{\ignore}[1]{}

\newcommand{\hcm}[1][1]{\hspace*{#1 cm}}

\newcommand{\wt}{\tilde{w}}
\newcommand{\symdiff}{\oplus}
\newcommand{\Vodd}{\mathcal{V}_{\operatorname{odd}}}

\newcommand{\BIP}{{\sc bipartite}}
\newcommand{\RAND}{{\sc randomized}}
\newcommand{\CARD}{{\sc cardinality only}}
\newcommand{\MCM}{{\sc mcm}}
\newcommand{\MCPM}{{\sc mcpm}}
\newcommand{\MCMs}{{\sc mcm}s}
\newcommand{\MWM}{{\sc mwm}}
\newcommand{\MWMs}{{\sc mwm}s}
\newcommand{\MWPM}{{\sc mwpm}}
\newcommand{\MWMONLY}{{\sc mwm only}}
\newcommand{\INT}{{\sc integer weights}}
\newcommand{\Gelig}{G_{elig}}
\newcommand{\scale}{\operatorname{scale}}
\newcommand{\msb}{\operatorname{MSB}}
\newcommand{\VODD}{V_{\operatorname{odd}}}
\newcommand{\VEVEN}{V_{\operatorname{even}}}

\newcommand{\citeyear}[1]{\cite{#1}}

\begin{document}

\title{Scaling Algorithms for Approximate\\ and Exact Maximum Weight Matching\thanks{This 
work is supported by NSF CAREER grant no. CCF-0746673
and a grant from the US-Israel Binational Science Foundation.
H.-H. Su is supported by a Taiwan (R.O.C.) Ministry of Education Fellowship.
Authors' emails: duanr02@gmail.com, pettie@umich.edu, hsinhao@umich.edu.}}

\author{Ran Duan\\
Max-Planck-Institut\\ f\"{u}r Informatik
\and
Seth Pettie\\
University of Michigan
\and
Hsin-Hao Su\\
University of Michigan
}

\maketitle

\begin{abstract}
The {\em maximum cardinality} and {\em maximum weight matching} problems can be solved in time 
$\tilde{O}(m\sqrt{n})$, a bound that has resisted improvement despite decades of research.
(Here $m$ and $n$ are the number of edges and vertices.)
In this article we demonstrate that this ``$m\sqrt{n}$ barrier'' is extremely fragile, in the following sense.
For any $\epsilon>0$, we give an algorithm that computes a $(1-\epsilon)$-approximate maximum weight matching
in $O(m\epsilon^{-1}\log\epsilon^{-1})$ time, that is, optimal {\em linear time} for any fixed $\epsilon$.
Our algorithm is dramatically simpler than the best exact maximum weight matching algorithms on general graphs
and should be appealing in all applications that can tolerate a negligible relative error.

Our second contribution is a new {\em exact} maximum weight matching algorithm for integer-weighted bipartite graphs
that runs in time $O(m\sqrt{n}\log N)$. 
This improves on the $O(Nm\sqrt{n})$-time and $O(m\sqrt{n}\log(nN))$-time algorithms known since the mid 1980s,
for $1\ll \log N \ll \log n$.  Here $N$ is the maximum integer edge weight.
\end{abstract}

\section{Introduction}\label{intro}

\nocite{Easterfield46}
\nocite{GT88}
\nocite{IbarraM81}
\nocite{Karzanov76}
\nocite{WitzgallZ65}
\nocite{Balinski69}
\nocite{J75}
\nocite{DH03a,DH03c}
\nocite{Gabow83}
\nocite{AltBMP91}
\nocite{FederM95}
\nocite{Gabow74}
\nocite{CunninghamM78}
\nocite{GGS89}
\nocite{GalilMG86}

Graph matching is one of the most well studied problems in combinatorial optimization.  The original motivations
of the problem were minimizing transportation costs~\cite{Hitchcock41,Kantorovitch42} and optimally assigning personnel to job positions~\cite{Easterfield46,Thorndike50}.
Over the years matching algorithms have found applications in scheduling, approximation algorithms, network switching, and as key
subroutines in other optimization algorithms, for example, undirected shortest paths~\cite{Lawler76}, planar max cut~\cite{OrlovaD72,Hadlock75}, Chinese postman tours~\cite{EdmondsJ73,Kwan62}, and metric traveling salesman~\cite{Christofides76}.
In most practical applications it is {\em not} critical that the algorithm produce an exactly optimum solution.  In this article we explore the extent to which this freedom---not demanding exact solutions---allows us to design simpler and more efficient algorithms.

In order to discuss prior work with precision we must introduce some notation and terminology.  The input is 
a weighted graph $G=(V,E,w)$ where $n=|V|$ and $m=|E|$ are the number of vertices and edges and $w$ is the edge weight function.
If $w$ assigns integer (rather than real) weights, let $N$ be the largest magnitude of a weight.  An unweighted graph is one for which $w(e)=1$ for all $e\in E$.
A {\em matching} is a set of vertex-disjoint edges and a {\em perfect} matching is one in which all vertices are matched.
The weight of a matching is the sum of its edge weights.
We use \MWM{} (and \MWPM) to denote the problem of finding a maximum weight (perfect) matching, as well as the matching itself.  We use \MCM{} and \MCPM{}
for the cardinality (unweighted) versions of these problems.  The \MWPM{} problem on bipartite graphs is often called the {\em assignment} problem. 

\begin{table}[t]
\caption{\label{tbl:cardhistory}{\sc Cardinality Matching}}
\centering
\begin{tabular}{|@{\hcm[.2]}c@{\hcm[.2]}|l@{\hcm[.2]}|l@{\istrut[1]{3.5}\hcm[.2]}|}
\multicolumn{3}{c}{}\\
\multicolumn{1}{l}{Year} & \multicolumn{1}{l}{Authors}	& \multicolumn{1}{l}{Time Bound \& Notes\hcm[2]}\\\cline{1-3}
		& folklore/trivial & $mn$	\hfill \BIP\\\cline{1-3}
1965 & Edmonds	& poly$(n)$	 \\\cline{1-3}
1965	& Witzgall \& Zahn	&		\\
1969	& Balinski		&		\\
1974	& Kameda \& Munro & 	 \\
1976	& Gabow		& 	\rb{2.5}{$mn$ or $mn\alpha(m,n)$ or $n^3$} \\
1976	& Lawler		&		\\
1976 & Karzanov	&	\\\cline{1-3}
1971	& Hopcroft \& Karp & \\
1973	& Dinic \& Karzanov	& \rb{2.5}{$m\sqrt{n}$}	 \hfill \rb{2.5}{\BIP}\\\cline{1-3}
1980	& Micali \& Vazirani & 		 \\
1991	& Gabow \& Tarjan	 & \rb{2.5}{$m\sqrt{n}$}	\\\cline{1-3}
1981 & Ibarra \& Moran	& $n^\omega$ \hcm \hfill \CARD,\RAND,\BIP\\\cline{1-3}
	&				& $n^\omega$ \hfill \CARD,\RAND\\
\rb{2.5}{1989}	& \rb{2.5}{Rabin \& Vazirani} & $n^{\omega+1}$ \hfill \RAND\\\cline{1-3}
1991 & Alt, Blum, Mehlhorn \& Paul \hcm[.3] & $n\sqrt{nm/\log n}$  \hfill \BIP\\\cline{1-3}
1991	& Feder \& Motwani	& \\
1997 & Goldberg \& Kennedy & \rb{2.5}{$m\sqrt{n}/\kappa$} \hfill \rb{2.5}{\BIP, $\kappa = \f{\log n}{\log(n^2/m)}$}\\\cline{1-3}
1996 & Cheriyan \& Mehlhorn & $n^2 + n^{5/2}/w$ \hfill \hfill \BIP, $w =$ machine word size\\\cline{1-3}
2004	& Goldberg \& Karzanov & $m\sqrt{n}/\kappa$ \\\cline{1-3}
2004	& Mucha \& Sankowski  &		                   \\
2006	& Harvey				&	\rb{2.5}{$n^\omega$} \hfill \rb{2.5}{\RAND}         \\\cline{1-3}
\multicolumn{3}{l}{\emph{Note:} Here $\omega<2.376$ is the exponent of $n\times n$ matrix multiplication.}
\end{tabular}

\end{table}

The \MWPM{} and \MWM{} problems are reducible to each other.
Given an instance $G$ of \MWM{}, let $G'$ consist of two copies of $G$ with zero-weight edges connecting copies of the same vertex.
Clearly a \MWPM{} in $G'$ corresponds to a pair of \MWMs{} in $G$.  In the reverse direction, if $G$ is an instance of \MWPM{} with weight function $w$,
find the \MWM{} of $G$ using the weight function $w'(e) = w(e) + nN$.  Maximum weight matchings with respect to $w'$ necessarily have maximum cardinality.
Call a matching {\em $\delta$-approximate}, where $\delta \in [0,1]$, if its weight is at least a factor $\delta$ of the optimum matching.
Let $\delta$-\MWM{} (and $\delta$-\MCM) be the problem of finding $\delta$-approximate maximum weight (cardinality) matching, as well as the matching itself.

Tables~\ref{tbl:cardhistory}, \ref{tbl:bipartitehistory}, and \ref{tbl:generalhistory} 
give an at-a-glance history of exact matching algorithms.  Algorithms are dated according to their initial publication, and
are included either because they establish a new time bound, or employ a noteworthy technique, or are of historical interest.  Table~\ref{tbl:approxhistory} gives a history of approximate \MCM{} and \MWM{} algorithms.

\subsection{Algorithms for Bipartite Graphs}

The \MWM{} problem is expressible as the following integer linear program, where $x$ represents the incidence vector of the matching.

\begin{align}
\mbox{maximize} \hcm[.5] & \sum_{e\in E} w(e)x(e)\nonumber\\
\mbox{subject to} \hcm[.5]  & 0\le x(e) \le 1 & \forall e\in E\label{LP}\\
&\sum_{e=(u,u')\in E} x(e) \le 1 &\forall u\in V\nonumber\\
&\istrut[0]{6}\mbox{$x(e)$ is an integer} & \forall e\in E\label{INT}
\intertext{It is well known that in {\em bipartite} graphs the integrality requirement~(\ref{INT}) is redundant,
that is, the basic feasible solutions of the LP~(\ref{LP}) are nonetheless integral.  See~\cite{Birkhoff46,Dantzig51}.
The dual of~(\ref{LP}) is}
\mbox{minimize} \hcm[.5] & \sum_{u\in V} y(u)\nonumber\\
\mbox{subject to} \hcm[.5] & y(e) \ge w(e) & \forall e\in E		\label{DUAL-BIP}\\
					& y(u) \ge 0 & \forall u\in V\nonumber\\
\mbox{where, by definition,} \hcm[.5] & y(u,v) \bydef y(u) + y(v)\nonumber
\end{align}

\begin{table}[t]
\caption{\label{tbl:bipartitehistory}{\sc Weighted Matching: Bipartite Graphs}}
\centering
\begin{tabular}{|c|l|l|@{\istrut[1]{3.5}}l|}
\multicolumn{3}{c}{}\\
\multicolumn{1}{l}{Year} & \multicolumn{1}{l}{Authors}	& \multicolumn{1}{l}{Time Bound \& Notes}\\\cline{1-3}
1946 & Easterfield & $2^n\mathrm{poly}(n)$\\\cline{1-3}
1953 & von Neumann &\\
1955 & Kuhn		& 	 \\
1955 & Gleyzal	& $\mathrm{poly}(n)$\\
1957 & Munkres	& 	 \\
1964 & Balinski \& Gomory & 	 \\	\cline{1-3}
1969 & Dinic \& Kronrod & $n^3$ \\\cline{1-3}
1970 & Edmonds \& Karp & \hfill $\mathrm{SP}^+ = $ time for one SSSP computation on\\
1971 & Tomizawa	& \rb{2.5}{$n\cdot \mathrm{SP}^+$} \hfill a non-negatively weighted graph\\\cline{1-3}
1975 & Johnson 	& $mn\log_d n$	 \hfill $d=2+m/n$\\\cline{1-3}
	 &  					&	$mn^{3/4}\log N$ \hfill \INT\\
\rb{2.5}{1983} & \rb{2.5}{Gabow} &	$Nm\sqrt{n}$ \hfill \MWMONLY, \INT\\\cline{1-3}
1984 & Fredman \& Tarjan & $mn + n^2\log n$   \\\cline{1-3}
1988 & Gabow \& Tarjan	&  \\
1992 & Orlin \& Ahuja & $m\sqrt{n}\log(nN)$  \hfill \INT\\
1997 & Goldberg \& Kennedy & \\\cline{1-3}
1996 & Cheriyan \& Mehlhorn & $n^{5/2}\log(nN)(\f{\log\log n}{\log n})^{1/4}$ \hfill \INT\\\cline{1-3}
	&		& $Nm\sqrt{n}/\kappa$ \hfill \MWMONLY, \INT\\
\rb{2.5}{1999} & \rb{2.5}{Kao, Lam, Sung \& Ting} & $N(n^2 + n^{5/2}/w)$ \hfill \MWMONLY, \INT\\\cline{1-3}
2004 & Mucha \& Sankowski & $Nn^\omega$ \hcm[1.5] \hfill \MWMONLY, \RAND, \INT\\\cline{1-3}
2006 & Sankowski 	& $Nn^\omega$ \hfill \RAND, \INT\\\cline{1-3}
\multicolumn{2}{|c|@{\istrut[1]{3.5}}}{} & $m\sqrt{n}\log N$ \hfill \MWMONLY, \INT 	\\
\multicolumn{2}{|c|@{\istrut[1]{3.5}}}{\rb{2.5}{\bf new}} & $m\sqrt{n}\log(nN)$	 \hfill \INT\\\cline{1-3}
\multicolumn{3}{l}{\parbox{6in}{\emph{Note:} $N$ is the maximum integer edge weight, $w$ is the machine word size,
and $\kappa = \log n/\log(n^2/m)$.
The time bounds of Johnson~\cite{J75} and Fredman and Tarjan~\cite{FT87} reflect faster priority queues.
The time bound of Mucha and Sankowski~\cite{MS04} follows from~Kao et al.'s~\citeyear{KaoLST01} reduction.}}
\end{tabular}

\end{table}

(In the \MWPM{} problem $\sum_{e=(u,u')} x(e) = 1$ holds with equality in the primal and $y(u)$ is unconstrained in the dual.)
Kuhn's~\citeyear{Kuhn55,Kuhn56} publication of the {\em Hungarian method} stimulated research on this problem from an 
{\em algorithmic} perspective, but it was not without precedent.  Kuhn noted that the algorithm was latent in the work
of Hungarian mathematicians K\"onig and \Egervary.\footnote{A translation of \Egervary's work appears in Kuhn~\citeyear{Kuhn55b}.}
However, the history goes back even further.  A recently rediscovered article
of Jacobi from 1865 describes a variant of the Hungarian algorithm; see~\cite{Ollivier09}.
Although Kuhn's algorithm self-evidently runs in polynomial time, this mark of efficiency was noted later: Munkres~\citeyear{Munkres57} 
showed that $O(n^4)$ time is sufficient.

Kuhn's Hungarian algorithm is sometimes described as a {\em dual} (rather than primal) algorithm, 
due to the fact that it maintains feasibility of the dual~(\ref{DUAL-BIP}) and progressively improves the primal objective~(\ref{LP})
by finding augmenting paths.
Gleyzal~\citeyear{Gleyzal55} (see also~\cite{BalinskiG64}) gave a {\em primal} algorithm for the assignment problem
in which the primal is feasible (the current matching is perfect) and the dual objective is progressively improved via weight-augmenting cycles.\footnote{The idea of {\em cycle canceling} is usually attributed to Robinson~\citeyear{Robinson49}.
Some assignment algorithms simply do not fit the primal/dual mold.  Von Neumann~\citeyear{vonNeumann53}, for example, gave a reduction from the assignment problem
to finding the optimum strategy in a zero-sum game given as an $n\times n^2$ matrix, which can be solved in polynomial time~\cite{BrownvN50}.}

The search for faster assignment algorithms began in earnest in the 1960s.  Dinic and Kronrod~\citeyear{DinicK69} gave an $O(n^3)$-time algorithm
and Edmonds and Karp~\citeyear{EdmondsK72} and Tomizawa~\citeyear{Tomizawa71} observed that assignment is reducible to $n$ single-source shortest path computations on a non-negatively weighted
directed graph.\footnote{It was known that the assignment problem is reducible to $n$ shortest path computations on arbitrarily weighted graphs.  
See Ford and Fulkerson~\citeyear{FF62}, Hoffman and Markowitz~\citeyear{HoffmanM63}, and Desler and Hakimi~\citeyear{DeslerH69} for different reductions.} 
Using Fibonacci heaps, $n$ executions of Dijkstra's~\citeyear{Dij59} shortest path algorithm take $O(mn+n^2\log n)$ time.
On integer weighted graphs this algorithm can be implemented slightly faster,
in $O(mn + n^2\log\log n)$ time~\cite{Han02,Tho03} or $O(mn)$ time (randomized)~\cite{AnderssonHNR98,Thorup07b},
independent of the maximum edge weight. 
Gabow and Tarjan~\citeyear{GT89}, improving an earlier algorithm of Gabow~\citeyear{Gab85}, 
gave a {\em scaling} algorithm for the assignment problem running in $O(m\sqrt{n}\log(nN))$ time, 
which is just a $\log(nN)$ factor slower than the fastest \MCM{} algorithm~\cite{HK73}.\footnote{Gabow and Tarjan's algorithm takes a Hungarian-type approach.
The same time bound has been achieved by Orlin and Ahuja~\citeyear{OrlinA92} using the {\em auction} approach of Bertsekas~\citeyear{Bertsekas81}, and by Goldberg and Kennedy~\citeyear{GoldbergK97} using a {\em preflow-push} approach.}
For reasonably sparse graphs Gabow and Tarjan's~\citeyear{GT89} assignment algorithm remains unimproved.  However, faster algorithms have
been developed when $N$ is small or the graph is dense~\cite{CheriyanM96,KaoLST01,Sankowski09}.
Of particular interest is Sankowski's algorithm~\citeyear{Sankowski09}, which solves \MWPM{} in $O(Nn^\omega)$ time, 
where $\omega$ is the exponent of square matrix multiplication.

\subsection{Algorithms for General Graphs}

Whereas the basic solutions to~(\ref{LP},\ref{DUAL-BIP}) are integral on bipartite graphs, the same is not true for general graphs.
For example, if the graph is a unit-weighted cycle with length $2k+1$ the \MWM{} has weight $k$ but (\ref{LP}) achieves its maximum of $k+1/2$ by setting $x(e)=1/2$ for all 
$e\in E$.  Let $\Vodd$ be the set of all odd-size subsets of $V$.  Clearly every feasible solution to the integer linear program
(\ref{LP},\ref{INT}) also satisfies the following odd-set constraints.
\begin{align}
& \sum_{e \in E(B)} x(e) \le (|B|-1)/2 & \forall B\in \Vodd\label{ODD}
\end{align}

\begin{table}[t]
\centering
\caption{\label{tbl:generalhistory}{\sc Weighted Matching: General Graphs}}
\begin{tabular}{|c|l|l|@{\istrut[1]{3.5}}}
\multicolumn{3}{c}{}\\
\multicolumn{1}{l}{Year} & \multicolumn{1}{l}{Authors}	& \multicolumn{1}{l}{Time Bound \& Notes}\\\cline{1-3}
1965 & Edmonds	& $\poly(n)$	 \\\cline{1-3}
1974 & Gabow	& 	 \\
1976 & Lawler		& \rb{2.5}{$n^3$}		\\\cline{1-3}
1976 & Karzanov	& $n^3 + mn\log n$		\\\cline{1-3}
1978 & Cunningham \& Marsh & $\poly(n)$		\\\cline{1-3}
1982 & Galil, Micali \& Gabow & $mn\log n$    \\\cline{1-3}
	 &  					&	$mn^{3/4}\log N$ \hfill \INT\\
\rb{2.5}{1985} & \rb{2.5}{Gabow} &	$Nm\sqrt{n}$ \hfill \MWMONLY, \INT\\\cline{1-3}
1989 & Gabow, Galil \& Spencer \hcm[.1] & $mn\log\log\log_d n + n^2\log n$ \hfill $d=2+m/n$\\\cline{1-3}
1990 & Gabow	& $mn + n^2\log n$	 \\\cline{1-3}
1991 & Gabow \& Tarjan	& $m\sqrt{n\log n}\log(nN)$  \hfill \INT\\\cline{1-3}
2006 & Sankowski 	& $Nn^\omega$ \hfill {\sc \underline{weight only}}, \INT\\\cline{1-3}
	&			&	$Nm\sqrt{n}/\kappa$		\hfill \MWMONLY, \INT\\
\rb{2.5}{2012} & \rb{2.5}{Huang \& Kavitha} & $Nn^\omega$ \hcm[.4] \hfill \MWMONLY, \RAND, \INT\\\cline{1-3} 
\multicolumn{3}{l}{\parbox{6in}{\emph{Note:} $N$ is the maximum integer edge weight, $\omega$ is the exponent of $n\times n$ matrix multiplication,
and $\kappa = \log n/\log(n^2/m)$.}}
\end{tabular}

\end{table}

Edmonds~\citeyear{Edmonds65,Ed65} proved that if we replace the integrality constraints~(\ref{INT}) with
(\ref{ODD}), the basic solutions to the resulting LP are integral.\footnote{In the \MWPM{} problem $\sum_{e=(u,u')} x(e) = 1$, for all $u\in V$,
and we have the freedom to use an alternative variety of odd-set constraints, namely, 
$\sum_{e=(u,v)\in E \::\: u\in B, v\not\in B} x(e) \ge 1, \; \forall B\in \Vodd$.}
Edmonds' algorithm mimics the structure of the Hungarian algorithm but the search for augmenting paths
is complicated by the presence of odd-length alternating cycles and the fact that matched edges must be searched in both directions.
Edmonds' solution is to contract {\em blossoms} as they are encountered.
A blossom is defined inductively as an odd-length cycle alternating between matched and unmatched edges,
whose components are either single vertices or blossoms in their own right.  Blossoms are discussed in detail in Section~\ref{sect:LP}.

The fastest implementation of Edmonds' algorithm, due to Gabow~\citeyear{G90}, runs in $O(mn + n^2\log n)$ time, 
which matches the running time of the best bipartite \MWPM{} algorithm~\cite{FT87}.
Gabow and Tarjan~\citeyear{GT91} extended their scaling algorithm for \MWPM{} to general graphs, achieving a running time
of $O(m\sqrt{n\log n}\log(nN))$, which is the fastest known algorithm for integer-weighted graphs and nearly matches the $O(m\sqrt{n})$ time bound of the best \MCM{} algorithms~\cite{MV80,Vazirani94}.\footnote{Gabow and Tarjan~\citeyear{GT91} claim a running time of $O(m\sqrt{n\log n\alpha(m,n)}\log(nN))$,
where the $\alpha(m,n)$ factor comes from an $O(m\alpha(m,n))$ implementation of the {\em split-findmin} data structure~\cite{G85}.  This can be reduced to $O(m\log\alpha(m,n))$~\cite{Pet05b}.  However, Thorup~\citeyear{Tho99} noted that split-findmin can be implemented in $O(m)$ time on integer-weighted graphs.}
As in the bipartite case, faster algorithms for \MWM{} and \MWPM{} are known when the graph is dense or $N$ is small.  
Sankowski~\citeyear{Sankowski09} noted that the {\em weight} of the \MWPM{} could be computed in $O(Nn^\omega)$ time;
however, it remains an open problem to adapt the cardinality matching algorithms of~\cite{MS04,Harvey09} to weighted graphs.
Huang and Kavitha~\citeyear{HuangK12}, generalizing \cite{KaoLST01}, proved that \MWM{} is reducible to $N$ \MCM{} computations,
which, by virtue of~\cite{GoldbergK04,MS04,Harvey09}, implies a new bound of $O(N \cdot \min\{n^\omega, m\sqrt{n}\log(n^2/m)/\log n\})$.  

\begin{table}[t]
\centering
\caption{\label{tbl:approxhistory}{\sc Approximate Maximum Cardinality/Weight Matching}}
\begin{tabular}{|c|l|l|l|@{\istrut[1]{3.5}}}
\multicolumn{3}{c}{}\\
\multicolumn{1}{l}{Year} & \multicolumn{1}{l}{Authors}	& \multicolumn{1}{l@{\hcm[.2]}}{Approx. Problem}	& \multicolumn{1}{l}{Time Bound \& Notes}	\\\cline{1-4}
1971	& Hopcroft \& Karp	&		& \\
1973	& Dinic \& Karzanov & \rb{2.5}{$(1-\epsilon)$-\MCM} & \rb{2.5}{$m\epsilon^{-1}$} \hfill \rb{2.5}{\BIP}\\\cline{1-4}
1980	& Micali \& Vazirani &		& 	 \\
1991	& Gabow \& Tarjan	 &	\rb{2.5}{$(1-\epsilon)$-\MCM}	& \rb{2.5}{$m\epsilon^{-1}$}  \\\cline{1-4}
\multicolumn{4}{l}{}\\\cline{1-4}
	& folklore/trivial	& $\f{1}{2}$-\MWM	& $m\log n$     \\\cline{1-4}
1988 & Gabow \& Tarjan	& $(1-\epsilon)$-\MWM & $m\sqrt{n}\log(n/\epsilon)$ \hfill \BIP\\\cline{1-4}
1991 & Gabow \& Tarjan	& $(1-\epsilon)$-\MWM & $m\sqrt{n\log n}\log(n/\epsilon)$\\\cline{1-4}
1999 & Preis		&			&		\\
2003 & Drake \& Hougardy \hcm[.2] & \rb{2.5}{$\f{1}{2}$-\MWM}	& \rb{2.5}{$m$}	 \\\cline{1-4}
2003 & Drake \& Hougardy & $(\f{2}{3} - \epsilon)$-\MWM	& $m\epsilon^{-1}$    \\\cline{1-4}
2004 & Pettie \& Sanders	& $(\f{2}{3} - \epsilon)$-\MWM & $m\log\epsilon^{-1}$\\\cline{1-4}
2010 & Duan \& Pettie		& 						&\\
2010 & Hanke \& Hougardy	& \rb{2.5}{$(\f{3}{4}-\epsilon)$-\MWM}	& \rb{2.5}{$m\log n \log\epsilon^{-1}$} \\\cline{1-4}
2010 & Hanke \& Hougardy  & $(\f{4}{5}-\epsilon)$-\MWM & $m\log^2 n \log\epsilon^{-1}$ \istrut[2]{3.5} \\\cline{1-4}
\multicolumn{2}{|c@{\hcm[.3]}|}{} 	& & $m\epsilon^{-1}\log\epsilon^{-1}$ \hcm[.3] \hfill {\sc arbitrary weights}\\
\multicolumn{2}{|c@{\hcm[.3]}|}{\rb{2.5}{\bf new}}  & \rb{2.5}{$(1-\epsilon)$-\MWM} & $m\epsilon^{-1}\log N$ \hfill \INT\\\cline{1-4}
\multicolumn{4}{l}{\parbox{6in}{\emph{Note:} $N$ is the maximum integer edge weight and $\epsilon > 0$ is arbitrary.}}
\end{tabular}
\end{table}

\subsection{Approximating Weighted Matching}

The approximate \MWM{} problem is remarkable in that it has been studied for decades, has practical applications,
and yet, as late as 1999, essentially nothing better than the greedy algorithm was known.\footnote{The greedy algorithm repeatedly includes the heaviest edge in the matching and removes all incident edges. Gabow and Tarjan~\citeyear{GT89,GT91} observed 
that by retaining the $O(\log(n/\epsilon))$ high-order bits of the edge weights, their exact scaling algorithms become $\tilde{O}(m\sqrt{n})$-time $(1-\epsilon)$-\MWM{} algorithms for bipartite and general graphs.}
Moreover, the $(1-\epsilon)$-\MCM{} problem had been solved satisfactorily in the early 1970s.
Although not stated as such, the $O(m\sqrt{n})$-time exact \MCM{} algorithms \cite{HK73,Dinic70,Karzanov73a,MV80} 
are actually $(1-\epsilon)$-\MCM{} algorithms running in $O(m\epsilon^{-1})$ time.
These algorithms are based on three observations (i) a maximal set of shortest augmenting paths can be found in linear time,
(ii) augmenting along such a set increases the length of the shortest augmenting path, and (iii) that after $k$ rounds of such augmentations the
resulting matching is a $(1-\f{1}{k+1})$-\MCM.

Preis~\citeyear{Preis99} gave a linear time $\f{1}{2}$-\MWM{} algorithm, which improves on the greedy algorithm's $O(m\log n)$ running time but
not its approximation guarantee.
Drake and Hougardy~\cite{VinkemeierH05}
presented the first linear time algorithm with an approximation guarantee greater than $1/2$.  Specifically,
they gave a $(\f{2}{3}-\epsilon)$-\MWM{} algorithm running in $O(m\epsilon^{-1})$ time, for any $\epsilon>0$.  The dependence on $\epsilon$ was later improved by 
Pettie and Sanders~\citeyear{PS04}.  These algorithms are based on a weighted version of Hopcroft and Karp's~\citeyear{HK73} argument, namely
that any matching whose weight-augmenting paths {\em and cycles} have at least $k$ unmatched edges is necessarily a $(1-\f{1}{k})$-\MWM.
Algorithms are presented in~\cite{DuanP10b,Hanke04,HankeH10} with different time/approximation tradeoffs: a $(\f{3}{4}-\epsilon)$-\MWM{} algorithm
running in time $O(m\log n\log\epsilon^{-1})$ and a $(\f{4}{5}-\epsilon)$-\MWM{} algorithm running in $O(m\log^2 n\log\epsilon^{-1})$ time.

\subsection{New Results}

We present the first $(1-\epsilon)$-\MWM{} algorithm that significantly improves on the $\tilde{O}(m\sqrt{n})$ running times of~\cite{GT89,GT91}.
Our algorithm runs in $O(m\epsilon^{-1}\log\epsilon^{-1})$ time on general graphs and $O(\min\{m\epsilon^{-1}\log\epsilon^{-1},m\epsilon^{-1}\log N\})$ time on integer-weighted general graphs.  This is optimal for any fixed $\epsilon$
and near-optimal as a function of $\epsilon$, given the state-of-the-art in \MCM{} algorithms.\footnote{Note that any $(1-\epsilon)$-\MWM{}
algorithm running in $O(f(\epsilon)m)$ time yields an {\em exact} \MCM{} algorithm running in $O(m\cdot (f(\epsilon) + \epsilon n))$ time, for any $\epsilon$. Thus,
any $(1-\epsilon)$-\MWM{} algorithm running in $o(m\epsilon^{-1})$ time would improve the $O(m\sqrt{n})$ \MCM{} algorithms~\cite{HK73,Dinic70,Karzanov73a,MV80}.}
Moreover, our algorithm is as {\em simple} as one could reasonably hope for.  Its search for augmenting paths uses depth first search~\cite[\S 8]{GT91} rather than the double depth first search of \cite{MV80}.
It uses no priority queues, split-findmin structures~\cite{G85}, or the blossom ``shells'' that arise from 
Gabow and Tarjan's~\citeyear{GT91} scaling technique.

Our second result is a new algorithm for exact \MWM{} on bipartite graphs running in $O(m\sqrt{n}\log N)$ time, which improves on \cite{Gab85,GT89}
for $1 \ll \log N \ll \log n$.  According to the \MWPM$\rightarrow$\MWM{} reduction, this also yields a new $O(m\sqrt{n}\log(nN))$ \MWPM{} algorithm, matching the performance of~\cite{GT89}.
However, our algorithm can be used to solve \MWPM{} directly, in $\log(\sqrt{n}N)$ scales rather than $\log(nN)$, which might be practically significant.
In terms of technique, our algorithm is a synthesis of the dual (Hungarian-type) approach of Gabow and Tarjan~\citeyear{GT89} and the primal approach of 
Balinski and Gomory~\citeyear{BalinskiG64}, among others.  The $\sqrt{n}$ factor in our running time arises not from the standard blocking flow-type argument~\cite{Karzanov73a,HK73}
but Dilworth's lemma~\citeyear{Dilworth50}, which ensures that every partial order on $n$ elements contains a chain or anti-chain with size $\sqrt{n}$.
Dilworth's lemma has also been used in Goldberg's~\citeyear{G95} single-source shortest path algorithm.

\subsection{Remarks on Approximate Weighted Matching and Its Applications}

Our focus is on algorithms that accept {\em arbitrary} input graphs and that give provably good {\em worst-case} approximations.
These twin objectives are self-evidently attractive, yet nearly all work (prior to~Preis~\citeyear{Preis99}) on approximate weighted matching
focused on specialized cases or weaker approximation guarantees.  Early work on the problem usually considered
complete bipartite graphs, and confirmed the efficiency of heuristics either experimentally or analytically 
with respect to inputs over some natural distribution~\cite{Brogden46,Thorndike50,Motzkin56,KuhnB62,Kurtzberg62,Avis78}.
See Avis~\citeyear{Avis83} for a more detailed discussion of heuristics.

Most work in the area considers graphs defined by metrics, often Euclidean metrics.
Reingold and Tarjan~\citeyear{ReingoldT81} proved that the greedy algorithm for metric \MWPM\footnote{For metric inputs
let \MWPM{} be the {\em minimum} weight perfect matching problem.} has an approximation ratio of $\approx n^{\log\f{3}{2}} > n^{0.58}$.
Goemans and Williamson~\citeyear{GW95} gave a 2-approximation for metric \MWPM{} that can be implemented in $O(n^2)$ time~\cite{GabowP02},
or $O(m\log^2 n)$ time~\cite{ColeHLP01} in metrics defined by $m$-edge graphs.  The Euclidean \MWPM{} comes in two flavors:
the monochromatic version is given $2n$ points and the bichromatic version is given $2n$ points, $n$ of which are colored blue, the rest red, where the matching
cannot include monochromatic edges.\footnote{The weight of the bichromatic \MWPM{} is also known as 
the {\em earth mover distance} between the red and blue points.}
Varadarajan and Agarwal~\citeyear{VaradA99} gave $(1+\epsilon)$-\MWPM{} algorithms for
the mono- and bichromatic variants running in time $O(n\poly(\epsilon^{-1}\log n))$ and $O(n^{3/2}\poly(\epsilon^{-1}\log n))$, respectively.
Other time-approximation tradeoffs for the bichromatic variant are possible~\cite{AgarwalV04,Charikar02}, including an $O(n\poly(\log n))$ time algorithm for $O(1)$-approximating the {\em weight} of the \MWPM{}~\cite{Indyk07}.
Some work considers the even more specialized case of Euclidean matching in the unit square, which allows for algorithms that guarantee
absolute upper bounds on the weight of the matching; see~\cite{IriMM83,ReingoldS83,Avis83} and the references therein.

There are several applications of \MWM{} (on general or bipartite graphs) in which one would gladly sacrifice matching quality for speed.
In input-queued switches packets are routed across a {\em switch fabric} from input to output ports.
In each cycle one partial permutation can be realized.  Existing algorithms for choosing these matchings,
such as iSLIP~\cite{McKeown99} and PIM~\cite{AndersonOST93}, guarantee $\fr{1}{2}$-\MCMs{}
and it has been shown~\cite{McKeownAW96,GiacconeLS05} that (approximate) \MWMs{}
have good throughput guarantees, where edge weights are based on queue-length.  See also~\cite{LeonardiMNM03,ShahK02,ShahGP02}.
Approximate \MWM{} algorithms are a component in several
multilevel graph clustering libraries.\footnote{E.g., METIS~\cite{KarypisK98b}, PARTY~\cite{PreisD97}, PT-SCOTCH~\cite{Pellegrini-SCOTCH08}
CHACO~\cite{HendricksonL95b}, JOSTLE~\cite{WalshawC07},
and KaFFPa/KaFFPaE~\cite{HoltgreweSS10,SandersS12}.} 
(PARTY, for example,
builds a hierarchical clustering by iteratively finding and contracting approximate \MWMs; see~\cite{PreisD97}.)
Approximate \MWM{} algorithms are used as a heuristic preprocessing step in several sparse linear
system solvers \cite{OlschowkaN96,DuffG02,SchenkWH07,HagemannS06}.  The goal is to permute the rows/columns to
maximize the weight on or near the main diagonal.

\subsection{Organization}

Section~\ref{sect:def} introduces some notation, states well known properties 
of augmenting paths and blossoms, 
and reviews Edmonds' optimality conditions for weighted matching.
In Section~\ref{sect:approxmwm} we present our $(1-\epsilon)$-\MWM{} algorithm
and in Section~\ref{sect:exact} we present exact algorithms for bipartite \MWM{} and \MWPM.

\section{Preliminaries}\label{sect:def}

We use $E(H)$ and $V(H)$ to refer to the edge and vertex sets of $H$ or the graph induced by $H$,
that is, $V(E')$ is the set of endpoints of $E'\subseteq E$ and $E(V')$ is the edge set of the graph induced by $V'\subseteq V$.  
A {\em matching} $M$ is a set of vertex-disjoint edges. Vertices not incident to an $M$ edge are {\em free}.
An alternating path (or cycle) is one whose edges alternate between $M$ and $E\backslash M$.
An alternating path $P$ is {\em augmenting} if $P$ begins and ends at free vertices,
that is, $M\symdiff P \bydef (M\backslash P)\cup (P\backslash M)$ is a matching with cardinality $|M\symdiff P| = |M|+1$.

When we seek $(1-\epsilon)$ approximate solutions, we can afford to scale and round edge weights to small integers.  To see this, observe that the weight of the \MWM{} is at least $w_{\max} = \max\{w(e) \;|\; e\in E(G)\}$.
It suffices to find a $(1-\epsilon/2)$-\MWM{} $M$ with respect to the weight function
$\wt(e) = \floor{w(e)/\gamma}$ where $\gamma = \epsilon\cdot w_{\max}/n$.  
Note that $w(e) - \gamma < \gamma\cdot \wt(e) \le w(e)$ for any $e$.
It follows from the definitions that:
\begin{align*}
w(M) &\ge \gamma\cdot \wt(M) 						& \mbox{Defn.~of $\wt$}\\
	&\ge \gamma \cdot (1-\epsilon/2)\wt(M^*) 				& \mbox{Defn.~of $M$, $M^*$ is the \MWM}\\
	&> (1-\epsilon/2)(w(M^*) - \gamma n/2)				& \mbox{Defn.~of $\wt$, $|M^*| \le n/2$}\\
	&= (1-\epsilon/2)(w(M^*) - \epsilon\cdot w_{\max}/2)	& \mbox{Defn.~of $\gamma$}\\
	&> (1-\epsilon)w(M^*)							& \mbox{Since $w(M^*) \ge w_{\max}$}
\end{align*}
Since it is better to use an exact \MWM{} algorithm when $\epsilon < 1/n$ \cite{G90,GT91}, we assume,
henceforth, that $w\::\: E\rightarrow \{1,2,\ldots,N\}$, where $N\le n^2$ is the maximum integer edge weight.  
For notational convenience we also assume that $N$ is a power of 2.

\subsection{Blossoms and the LP Formulation of MWM}\label{sect:LP}

The dual LP of (\ref{LP},\ref{ODD}) is

\begin{align*}
\mbox{minimize} \hcm[.5] & \sum_{u\in V(G)} y(u) + \sum_{B\in \Vodd} \f{|B|-1}{2}\cdot z(B)\\
\mbox{subject to} \hcm[.5] & yz(e) \ge w(e) & \forall e\in E(G)\\
& y(u) \ge 0, z(B) \ge 0 & \forall u\in V(G), \forall B\in \Vodd\istrut[3]{0}\\
\mbox{where, by definition,} \hcm[.5] & 
yz(u,v) \bydef y(u) + y(v) + \sum_{\substack{B\in \Vodd,\\(u,v)\in E(B)}} \zero{z(B)}
\end{align*}

Despite the exponential number of primal constraints and dual $z$-variables, Edmonds demonstrated
that an optimum matching could be found in polynomial time without maintaining information ($z$-values)
on more than $n/2$ elements of $\Vodd$ at any given time.
At intermediate stages of Edmonds' algorithm there is a matching $M$
and a laminar (nested) subset $\Omega\subseteq \Vodd$, where each element of $\Omega$
is identified with a {\em blossom}.
Blossoms are formed inductively as follows.  
If $v\in V$ then the set $\{v\}$ is a trivial blossom.
An odd length sequence $(A_0,A_1,\ldots,A_\ell)$ forms a nontrivial blossom $B=\bigcup_i A_i$
if the $\{A_i\}$ are blossoms
and there is a sequence of edges $e_0,\ldots, e_\ell$ where $e_i\in A_i\times A_{i+1}$ (modulo $\ell+1$) and $e_i\in M$ if and only if $i$ is odd, 
that is, $A_0$ is incident to unmatched edges $e_0,e_{\ell}$.  See Figure~\ref{fig:blossom}.
The {\em base} of blossom $B$ is the base of $A_0$; the base of a trivial blossom is its only vertex.
The set of {\em blossom edges} $E_B$ are $\{e_0,\ldots,e_\ell\}$ and those used in the formation of $A_0,\ldots,A_\ell$.
The set $E(B) = E\cap (B\times B)$ may, of course, include many non-blossom edges.
A short proof by induction shows that $|B|$ is odd and that the base of $B$ is the only unmatched vertex in the subgraph induced by $B$.

Matching algorithms represent a nested set $\Omega$ of {\em active} blossoms by rooted trees,
where leaves represent vertices and internal nodes represent nontrivial blossoms.
A {\em root blossom} is one not contained in any other blossom.  The children of an internal node representing $B$ are 
ordered according to the odd cycle that formed $B$, where one child is distinguished as containing
the base of $B$.
As we will see, it is often possible to treat blossoms as if they were single vertices.
Let the {\em contracted graph} $G/\Omega$ be obtained by contracting all root blossoms and removing spurious edges.
To {\em dissolve} a root blossom $B$ means
to delete its node in the blossom forest and, in the contracted graph, to replace $B$ with individual vertices $A_0,\ldots,A_\ell$.
Lemma~\ref{lem:blossom} summarizes some useful properties of the contracted graph.

\begin{lemma}\label{lem:blossom}
Let $\Omega$ be a set of blossoms with respect to a matching $M$.
\begin{enumerate}
\item If $M$ is a matching in $G$ then $M/\Omega$ is a matching in $G/\Omega$.
\item Every augmenting path $P'$ relative to $M/\Omega$ in $G/\Omega$ extends to an augmenting path $P$ relative to $M$ in $G$.
(That is, $P$ is obtained from $P'$ by substituting for each non-trivial blossom vertex $B$ in $P'$ a path through $E_B$.  See Figure~\ref{fig:blossom}(a,b).)
\item If $P$ is an augmenting path and $P/\Omega$ is also an augmenting path relative to $M/\Omega$, then $\Omega$ remains a valid set of blossoms (possibly with different bases) for the augmented matching $M\symdiff P$.
See Figure~\ref{fig:blossom}(a,b).
\item The base $u$ of a blossom $B\in \Omega$ uniquely determines a maximum cardinality matching of $E_B$, having size $(|B|-1)/2$.  See Figure~\ref{fig:blossom}(a,b).
\end{enumerate}
\end{lemma}

Implementations of Edmonds' algorithm grow a matching $M$ while maintaining Property~\ref{prop:yz-strict}, which
controls the relationship between $M$, $\Omega$ and the dual variables.

\begin{property}\label{prop:yz-strict} (Complementary Slackness Conditions)
\begin{enumerate}
    \item {\em Nonnegativity:} $z(B)\geq 0$ for all $B\in\Vodd$ and $y(u)\ge 0$ for all $u\in V(G)$.
    \item {\em Active Blossoms:} $\Omega$ contains all $B$ with $z(B)>0$ and all root blossoms $B$ have $z(B)>0$. (Non-root blossoms may have zero $z$-values.)\label{item:z-strict}
    \item {\em Domination:} $yz(e)\geq w(e)$ for all $e\in E$.\label{item:yz-lb-strict}
    \item {\em Tightness:} $yz(e)= w(e)$ when $e\in M$ or $e\in E_B$ for some $B\in \Omega$.\label{item:yz-ub-strict}
\end{enumerate}
\end{property}

If the $y$-values of free vertices become zero, it follows from domination and tightness that $M$ is a maximum weight matching,
as the following short proof attests.  Here $M^*$ is any maximum weight matching.

\begin{align*}
w(M) &= \sum_{e\in M} w(e)\\
	&= \sum_{e\in M} yz(e) 		& \mbox{tightness}\\
        &= \sum_{u\in V(G)} y(u) + \sum_{B\in \Omega} \f{|B|-1}{2}\cdot z(B) & \mbox{Note } \sum_{u\in V(G)} y(u) = \sum_{u\in V(M)} y(u)\\
        &\ge \sum_{u\in V(M^*)} y(u) + \sum_{B\in \Omega} |E(B)\cap M^*|\cdot z(B)   & \mbox{$y,z$ non-negative}\\
        &= \sum_{e\in M^*} yz(e) \; \ge \; w(M^*) & \mbox{domination}
\end{align*}

\begin{figure}[t]
\centering
\begin{tabular}{c@{\hcm[1.5]}c}
\scalebox{.6}{\includegraphics{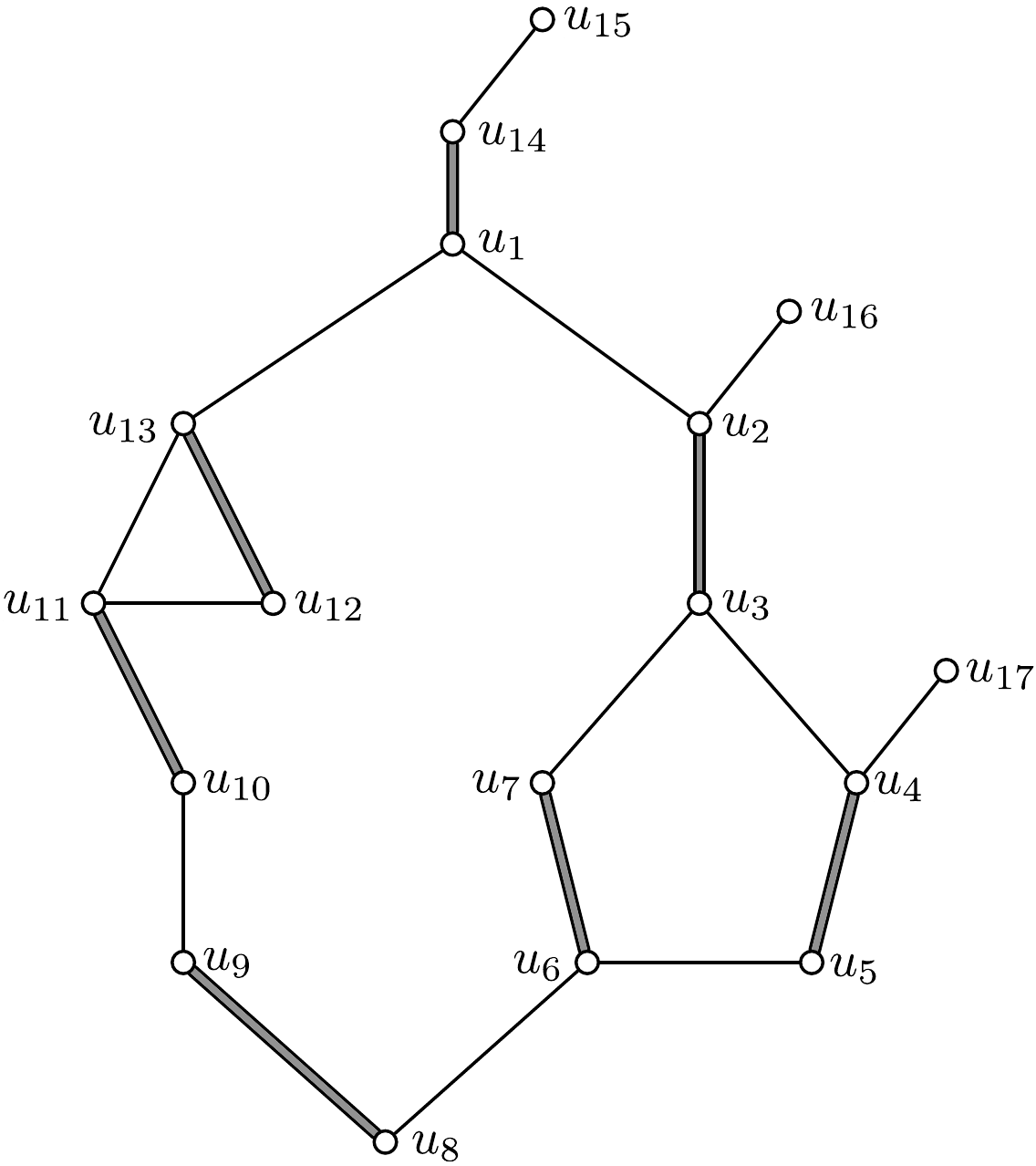}}
&
\scalebox{.6}{\includegraphics{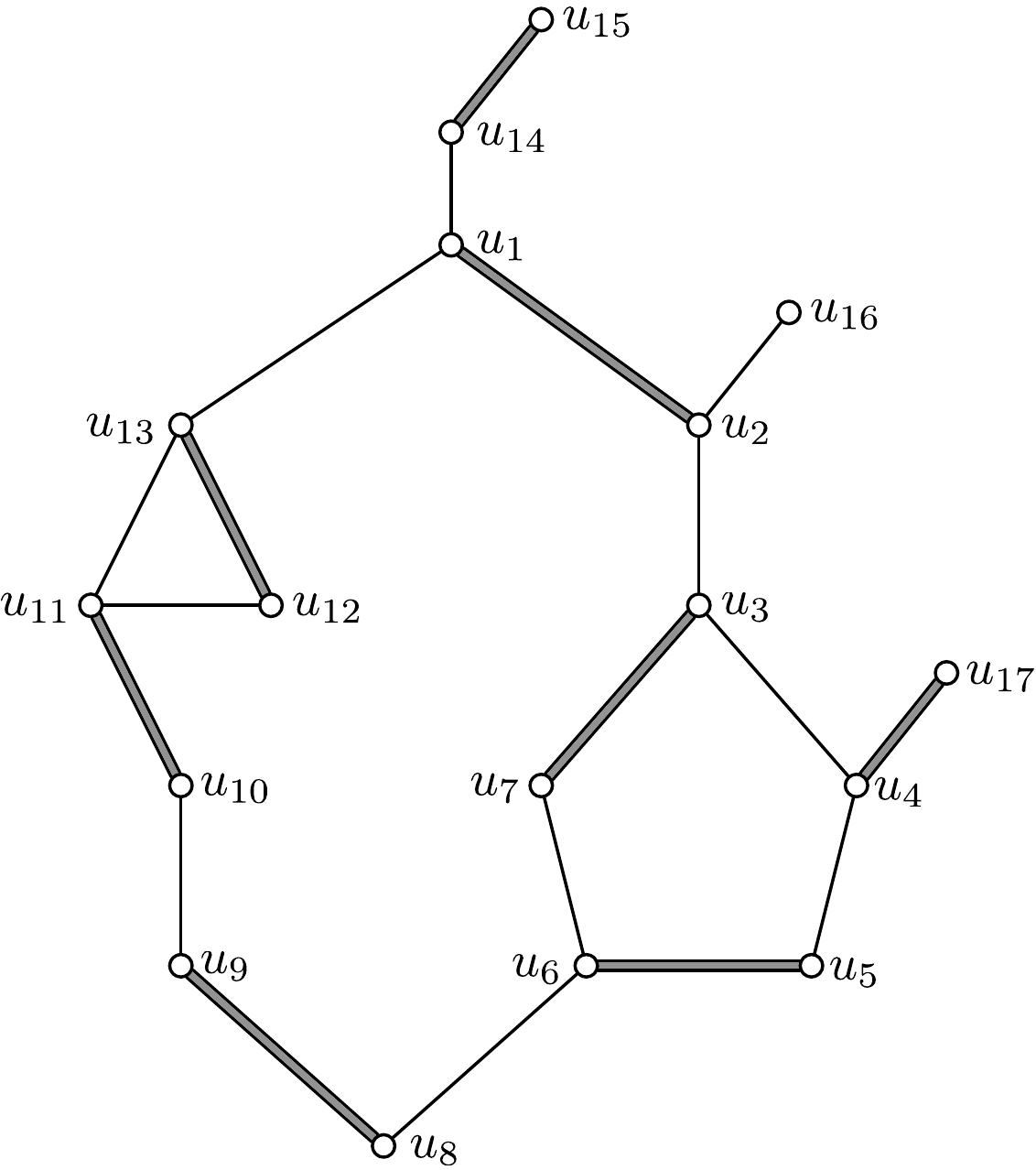}}\\
\ \\
(a) & (b)
\end{tabular}
\caption{\label{fig:blossom}Thick edges are matched, thin unmatched.
(a)
A blossom $B_1 = (u_1,u_2,B_2,u_8,u_9,u_{10},B_3)$ with base $u_1$
containing
non-trivial sub-blossoms $B_2 = (u_3,u_4,u_5,u_6,u_7)$ with base $u_3$
and $B_3 = (u_{11},u_{12},u_{13})$ with base $u_{11}$.  Vertices $u_{15},u_{16},$ and $u_{17}$ are free.
The path $(u_{16},u_2,u_3,u_7,u_6,u_5,u_4,u_{17})$ is an example of an augmenting path that 
exists in $G$ but not $G/B_1$,
the graph obtained by contracting $B_1$.
(b) 
The situation after augmenting along $(u_{15},u_{14},B_1,u_{17})$ in $G/B_1$,
which corresponds to augmenting along 
$(u_{15},u_{14},u_1,u_2,u_3,u_7,u_6,u_5,u_4,u_{17})$ in $G$.
After augmentation $B_1$ and $B_2$ have their base at $u_4$.
}
\end{figure}

\section{A Scaling Algorithm for Approximate \MWM}\label{sect:approxmwm}

Our algorithm maintains a {\em dynamic} relaxation of complementary slackness.  In the beginning {\em domination}
is weak but becomes progressively tighter at each scale whereas {\em tightness} is weakened at each scale,
though not uniformly.  The degree to which a matched edge or blossom edge may violate tightness depends on {\em when} it last entered the blossom or matching.

Recall that $N$ is the maximum integer edge weight.  The parameter $\epsilon' = \Theta(\epsilon)$ will be selected later to guarantee
that the final matching is a $(1-\epsilon)$-\MWM.  Henceforth, assume that $N\ge 1$ and $\epsilon' \le 1/4$ are powers of two.
Define $\delta_0 = \epsilon' N$ and $\delta_{i} = \delta_0/2^i$.
At scale $i$ we use the truncated weight function $w_i(e) = \delta_i \floor{w(e)/\delta_i}$.
Note that $w_{i+1}(e) = w_i(e)$ or $w_i(e)+\delta_{i+1}$.

\begin{property}\label{prop:yz} {(Relaxed Complementary Slackness)}
There are $L+1$ scales numbered $0,\ldots,L$, where $L = \log N$.
Let $i\in [0,L]$ be the current scale.

\begin{enumerate}
    \item {\em Granularity:} $z(B)$ is a nonnegative multiple of $\delta_i$, for all $B\in\Vodd$, and $y(u)$ is a nonnegative multiple of $\delta_i/2$, 
    for all $u\in V(G)$.\label{item:nonneg}
    \item {\em Active Blossoms:} $\Omega$ contains all $B$ with $z(B)>0$ and all root blossoms $B$ have $z(B)>0$. (Non-root blossoms may have zero $z$-values.)\label{item:z}
    \item {\em Near Domination:} $yz(e) \geq w_i(e) - \delta_i$ for all $e\in E$.\label{item:yz-lb}
    \item {\em Near Tightness:} Call a matched or blossom edge {\em type $j$} if it was last made a matched or blossom edge in scale $j\le i$.
    (That is, it entered the set $M\cup\bigcup_{B\in\Omega} E_B$ in scale $j$ and has remained in that set, even as $M$ and $\Omega$ evolve as augmenting paths are found and blossoms are formed and dissolved.)
    If $e$ is such a type $j$ edge then $yz(e) \le w_i(e) + 2(\delta_j - \delta_i)$.\label{item:yz-ub}
    \item {\em Free Vertex Duals:} The $y$-values of free vertices are equal and strictly less than the $y$-values of matched vertices.\label{item:y}
\end{enumerate}
\end{property}

Lemma~\ref{lem:approx} allows us to measure the quality of a matching $M$, given duals $y$ and $z$
satisfying Property~\ref{prop:yz}.

\begin{lemma}\label{lem:approx}
Let $M$ be a matching satisfying Property~\ref{prop:yz} at scale $i$ and let $M^*$ be a maximum weight matching. 
Let $f$ be the number of free vertices, each having $y$-value $\phi$, and let $\hat{\epsilon}$ be such that
$yz(e) - w(e) \le \hat{\epsilon}\cdot w(e)$ for all $e\in M$.
Then $w(M) \ge (1+\hat{\epsilon})^{-1}(w(M^*) - 2\delta_i  |M^*| - f\phi)$.  
If $i = L$ and $\phi=0$ then $M$ is a $(1-\epsilon' -\hat{\epsilon})$-\MWM.
\end{lemma}

\begin{proof}
The claim follows from Property~\ref{prop:yz}.
\begin{align}
w(M) &= \sum_{e\in M} w(e)						& \mbox{defn. of $w(M)$}\nonumber\\
	&\ge (1+\hat{\epsilon})^{-1}\sum_{e\in M} yz(e) 		& \mbox{near tightness, defn. of $\hat{\epsilon}$}\nonumber\\
        &= (1+\hat{\epsilon})^{-1}\paren{\sum_{u\in V(M)} y(u) + \sum_{B\in \Omega} \f{|B|-1}{2}\cdot z(B)} & \mbox{defn. of $yz$}\nonumber\\
        &\ge (1+\hat{\epsilon})^{-1}\paren{\sum_{u\in V(M^*)} y(u)  + \sum_{B\in \Omega} |E(B)\cap M^*|\cdot z(B) - f\phi}\hcm[-1]\label{eqn:approx1}\\
        &\ge (1+\hat{\epsilon})^{-1}\paren{\sum_{e\in M^*} yz(e) - f\phi} 		& \mbox{defn. of $yz$}\nonumber\\
        & \ge (1+\hat{\epsilon})^{-1}\paren{w_i(M^*)  - f\phi- \delta_i\cdot |M^*|} & \mbox{near domination}\nonumber\\
        & > (1+\hat{\epsilon})^{-1}\paren{w(M^*)  - f\phi- 2\delta_i\cdot |M^*|} & \mbox{defn. of $w_i$}\nonumber
\end{align}
Inequality~(\ref{eqn:approx1}) follows from several facts, namely, that no matching can contain more than $(|B|-1)/2$ edges in $B$,
that $V(M^*)\backslash V(M)$ contains only free vertices (with respect to $M$), whose $y$-values are $\phi$, and that $y$- and $z$-values are nonnegative.  Note that the last inequality is loose by $\delta_i|M^*|$ if $i=L$ since in that case $w_L = w$.

The integrality of edge weights implies that $w(M^*) \ge |M^*|$.  If $i=L$ and $\phi=0$ then
$\delta_L = \epsilon'N/2^L = \epsilon'$
and 
$w(M) \ge (1+\hat{\epsilon})^{-1}(w(M^*) - \delta_L|M^*|) \ge (1+\hat{\epsilon})^{-1}(1-\epsilon')w(M^*) > (1-\epsilon'-\hat{\epsilon})w(M^*)$,
that is, $M$ is a $(1-\epsilon'-\hat{\epsilon})$-\MWM.
\end{proof}

While not suggesting an algorithm per se, Lemma~\ref{lem:approx} tells us which invariants our algorithm must maintain and 
when it may halt with a $(1-\epsilon)$-\MWM.  
For example, if at the last scale $L$ we have $\delta_L \le \epsilon/2$ and for any type $j$ edge $e\in M$, $\delta_j \le (\epsilon/4)w(e)$ (i.e., $\hat{\epsilon} = \epsilon/2$) then as soon as $y$-values at free vertices reach zero, the current matching must be a $(1-\epsilon)$-\MWM.

\subsection{The Scaling Algorithm}

Initially $M=\emptyset, \Omega = \emptyset$, and
$y(u) = N/2 - \delta_0/2$ for all $u\in V$, which clearly satisfies Property~\ref{prop:yz} for scale $i=0$,
since $yz(e) = 2(N/2 - \delta_0/2) \ge w_0(e) - \delta_0$.
The algorithm, given in Figure~\ref{fig:alg}, consists of scales $0,\ldots,L=\log N$, where the purpose of
scale $i$ is to halve the $y$-values of free vertices while maintaining Property~\ref{prop:yz}.
In each iteration of scale $i$ the algorithm (1) augments a maximal set of augmenting paths of {\em eligible} edges,
(2) finds and contracts blossoms of {\em eligible} edges, (3) performs dual adjustments on $y$- and $z$-values, and (4) dissolves previously contracted
root blossoms if their $z$-values become zero.  
Each dual adjustment step decrements by $\delta_i/2$ the $y$-values of free vertices.  Thus, 
there are roughly $(N/2^i)/(\delta_i/2) = O(\epsilon^{-1})$ iterations per scale, independent of $i$.
The efficiency and correctness of the algorithm depend on {\em eligibility} being defined properly.

\begin{definition}\label{def:eligible}
At scale $i$, an edge $e$ is {\em eligible} if at least one of the following hold:

\parbox{5in}{
\begin{enumerate}
\item[(i)] $e\in E_B$ for some $B\in \Omega$.
\item[(ii)] $e\not\in M$ and $yz(e) = w_i(e) - \delta_i$.
\item[(iii)] $e\in M$ and $yz(e) - w_i(e)$  is a nonnegative integer multiple of $\delta_i$.
\end{enumerate}
}

Let $E_{elig}$ be the set of eligible edges 
and let $\Gelig = (V,E_{elig})/\Omega$ be the unweighted graph obtained by discarding ineligible edges and contracting
root blossoms.
\end{definition}

Criterion (i) for eligibility simply ensures that an augmenting path in $\Gelig$ extends to an augmenting path of eligible edges in $G$.
A key implication of Criteria (ii) and (iii) is that if $P$ is an augmenting path in $\Gelig$, every edge in $P$ becomes ineligible in $(M/\Omega) \symdiff P$.
This follows from the fact that unmatched edges must have $yz(e) - w_i(e) < 0$ whereas matched edges must have $yz(e) - w_i(e) \ge 0$.
Regarding Criterion (iii), note that Property~\ref{prop:yz} (granularity and near domination) implies that 
$yz(e) - w_i(e)$ is at least 
$-\delta_i$ and an integer multiple of $\delta_i/2$.

\begin{figure}
\centering
\framebox{
\begin{minipage}{6.2in}
\underline{Initialization:}
\begin{align*}
M &\leftarrow \emptyset & \zero{\hcm[-6]\mbox{no matched edges}}\\
\Omega &\leftarrow \emptyset &\zero{\hcm[-6]\mbox{no blossoms}}\\
\delta_0 &\leftarrow \epsilon' N & \zero{\hcm[-6]\mbox{$\epsilon' = \Theta(\epsilon)$; w.l.o.g., $\epsilon',N$ are powers of 2}}\\
y(u) &\leftarrow \f{N}{2}- \f{\delta_0}{2}, \mbox{ for all } u\in V(G) \hcm[7] & \zero{\hcm[-6]\mbox{satisfies Property~\ref{prop:yz}(\ref{item:yz-lb})}}\\
\end{align*}

Execute scales $i=0\ldots,L=\log N$ and return the matching $M$. \\

\underline{Scale $i$:}\\

\begin{itemize}
\item[] \begin{itemize}
	\item Repeat the following steps until $y$-values of free vertices reach $N/2^{i+2} - \delta_i/2$, if $i\in[0,L)$,
	or until they reach zero, if $i = L$.\\

    \begin{itemize}
    \item[(1)] {\bf Augmentation:}\\
    Find a maximal set $\Psi$ of augmenting paths in
    $\Gelig$ and set $M \leftarrow M\symdiff(\bigcup_{P\in\Psi}P)$.
    Update $\Gelig$.\\
    
    \item[(2)] {\bf Blossom Shrinking:}\\
    Let $V_{out}\subseteq V(\Gelig)$ be the vertices (that is, root blossoms) reachable from free vertices by even-length alternating paths; let $\Omega'$ be a maximal set of (nested) blossoms on $V_{out}$.  (That is, if $(u,v)\in E(\Gelig)\backslash M$ and $u,v\in V_{out}$, then $u$ and $v$ must be in a common blossom.)  
    Let $V_{in} \subseteq V(\Gelig)\backslash V_{out}$ be those non-$V_{out}$-vertices reachable from free vertices by odd-length alternating paths.
    	Set $z(B) \leftarrow 0$ for $B\in \Omega'$ and set $\Omega \leftarrow \Omega \cup \Omega'$.  Update $\Gelig$.\\
	
    \item[(3)] {\bf Dual Adjustment:}\\
    Let $\hat{V}_{in},\hat{V}_{out}\subseteq V$ be original vertices represented by vertices in $V_{in}$ and $V_{out}$.
    The $y$- and $z$-values for some vertices and root blossoms are adjusted:
    \begin{align*}
        y(u) &\leftarrow y(u)-\delta_i/2, \mbox{ for all $u\in \hat{V}_{out}$.}\\
        y(u) &\leftarrow y(u)+\delta_i/2,  \mbox{ for all $u\in \hat{V}_{in}$.}\\
        z(B) &\leftarrow z(B)+\delta_i, \mbox{ if $B\in\Omega$ is a root blossom with $B\subseteq \hat{V}_{out}$.}\\
        z(B) &\leftarrow z(B)-\delta_i, \mbox{ if $B\in\Omega$ is a root blossom with $B\subseteq \hat{V}_{in}$.}\hcm[2]
    \end{align*}

    \item[(4)] After dual adjustments some root blossoms may have zero $z$-values.  Dissolve such blossoms (remove them from $\Omega$) as long as they exist.
    Note that non-root blossoms are allowed to have zero $z$-values.  Update $\Gelig$ by the new $\Omega$.\\
\end{itemize}

\item Prepare for the next scale, if $i\in [0,L)$:
\begin{align*}
\delta_{i+1} &\leftarrow \delta_i/2\\
y(u) &\leftarrow y(u) + \delta_{i+1}, \mbox{ for all } u\in V(G).\hcm[5]
\end{align*}
\end{itemize}
\end{itemize}
\end{minipage}
}
\caption{\label{fig:alg} A $(1-\epsilon)$-approximate \MWM{} algorithm}
\end{figure}

\subsection{Analysis and Correctness}

Lemma~\ref{lem:propone-dualadj} states that the algorithm maintains 
Property~\ref{prop:yz} after each of the $O(\epsilon^{-1})$ Dual Adjustment steps in each scale.
Lemma~\ref{lem:non-exist-path} establishes the critical fact 
that all augmenting paths and paths from free vertices to blossoms are eliminated from $\Gelig$
before each Dual Adjustment step.

\begin{lemma}\label{lem:non-exist-path}
After the Augmentation and Blossom Shrinking steps $\Gelig$ contains
no augmenting path, nor is there a path from a free vertex to a blossom.
\end{lemma}

\begin{proof}
Suppose there is an augmenting path $P$ in $\Gelig$ after
augmenting along paths in $\Psi$.  Since $\Psi$ is maximal, $P$
must intersect some $P'\in\Psi$ at a vertex $v$.
However, after the Augmentation step every edge in $P'$ will become ineligible, so the
matching edge $(v,v')\in M$ is no longer in $\Gelig$, contradicting the fact that $P$ consists of eligible edges.
Since $\Omega'$ is maximal there can be no blossom reachable from a free vertex in $\Gelig$ after the Blossom Shrinking step.
\end{proof}

Lemma~\ref{lem:parity} guarantees that all $y$-values updated in a Dual Adjustment step
have the same parity as a multiple of $\delta_i/2$.  
In the proof of Lemma~\ref{lem:propone-dualadj} this fact is used to argue that if both endpoints of 
an edge $e$ have their $y$-values decremented, then $yz(e)$ is a multiple of $\delta_i$.

\begin{lemma}\label{lem:parity}
Let $R\subseteq V(\Gelig)$ be the set of vertices
reachable from free vertices by eligible alternating paths, at any point in scale $i$.  
Let $\hat{R} \subseteq V(G)$ be the set of original vertices represented by those in $R$.
Then the $y$-values of $\hat R$-vertices have the same parity, as a multiple of $\delta_i/2$.
\end{lemma}

\begin{proof}
Assume, inductively, that before the Blossom Shrinking step, all vertices in a common blossom have the same parity, as a multiple of $\delta_i/2$.  Consider an eligible path $P = (B_0,B_1,\ldots,B_k)$ in $\Gelig$, where the $\{B_j\}$ are either vertices or blossoms in $\Omega$ and $B_0$ is unmatched in $\Gelig$.  Let $(u_0,v_1), (u_1,v_2), \ldots, (u_{k-1},v_k)$ be the $G$-edges corresponding to $P$, where $u_j,v_j\in B_j$.  By the inductive hypothesis, $u_j$ and $v_j$ have the same parity, and whether $(u_j,v_{j+1})$ is matched or unmatched, Definition~\ref{def:eligible} stipulates that $yz(u_j,v_{j+1})/\delta_i$ is an integer, which implies $y(u_j)$ and $y(v_{j+1})$ have the same parity as a multiple of $\delta_i/2$.
Thus, the $y$-values of all vertices in $B_0\cup \cdots \cup B_k$ have the same parity as a free vertex in $B_0$, whose
$y$-value is equal to every other free vertex, by Property~\ref{prop:yz}(\ref{item:y}).
Since new blossoms are formed by eligible edges, the inductive hypothesis is preserved after the Blossom Shrinking step.  It is also preserved after the Dual Adjustment step since the $y$-values of vertices in a common blossom are incremented or decremented in lockstep.  This concludes the induction.
\end{proof}

\begin{lemma}\label{lem:propone-dualadj}
The algorithm preserves Property~\ref{prop:yz}.
\end{lemma}

\begin{proof}
Property~\ref{prop:yz}(\ref{item:y}) (free vertex duals) is obviously maintained as only free vertices have their $y$-values decremented in each Dual Adjustment step.
Property~\ref{prop:yz}(\ref{item:z}) (active blossoms) is also maintained since all the new root blossoms discovered in the Blossom Shrinking step are contained in $V_{out}$ and will have positive $z$-values after adjustment.
Furthermore, each root blossom whose $z$-value drops to zero is dissolved, after Dual Adjustment.
At the beginning of scale $i$ all $y$- and $z$-values are integer multiples of $\delta_i/2$ and $\delta_i$, respectively,
satisfying Property~\ref{prop:yz}(\ref{item:nonneg}) (granularity).  This property is clearly maintained in each Dual Adjustment step.  If $e\not\in M$ is placed in $M$ during an Augmentation step
or placed in $\bigcup_{B\in \Omega} E_B$ during a Blossom Shrinking step then 
$e$ is type $i$ and $yz(e) = w_i(e) - \delta_i < w_i(e)$, which satisfies Property~\ref{prop:yz}(\ref{item:yz-ub}).

It remains to show that the algorithm maintains Property~\ref{prop:yz}(\ref{item:yz-lb},\ref{item:yz-ub}) (near domination and near tightness).  First consider the dual adjustments made at the {\em end} of scale $i$ (the last line of pseudocode in Figure~\ref{fig:alg}.)
Let $e=(u,v)$ be an arbitrary edge and let $yz$ and $yz'$ be the function before and after dual adjustment.
It follows that
\begin{align*}
yz'(e) &= yz(e) + 2\delta_{i+1}	& \mbox{$y(u),y(v)$ incremented by $\delta_{i+1}$}\\
	&\ge w_i(e) - \delta_i + 2\delta_{i+1} & \mbox{near domination at scale $i$}\\
	&\ge w_{i+1}(e) - \delta_{i+1}		& \mbox{$w_i(e) \ge w_{i+1}(e) - \delta_{i+1}$}
\end{align*}
That is, Property~\ref{prop:yz}(\ref{item:yz-lb}) is preserved.
If $e\in M \cup \bigcup_{B\in \Omega} E_B$ is a type $j$ edge, then at the end of scale $i$
Property~\ref{prop:yz}(\ref{item:yz-ub}) is also preserved since
\[
yz'(e) = yz(e) + 2\delta_{i+1} \le w_i(e) + 2(\delta_j - \delta_i) + 2\delta_{i+1} \le w_{i+1}(e) + 2(\delta_j - \delta_{i+1})
\]
The first inequality follows from Property~\ref{prop:yz}(\ref{item:yz-ub}) at scale $i$ and the second
inequality from the fact that $w_i(e) \le w_{i+1}(e)$ and $\delta_i = 2\delta_{i+1}$.

Now consider a Dual Adjustment step.
If neither $u$ nor $v$ is in $\hat{V}_{in}\cup \hat{V}_{out}$ 
or if $u,v$ are in the same root blossom in $\Omega$, then $yz(e)$ is unchanged, preserving Property~\ref{prop:yz}.
The remaining cases depend on whether $(u,v)$ is in $M$ or not, whether $(u,v)$ is eligible or not, and whether
both $u,v \in \hat{V}_{in}\cup \hat{V}_{out}$ or not.

\paragraph{Case 1: $e \not\in M, \, u,v \in \hat{V}_{in}\cup \hat{V}_{out}$}
If $e$ is ineligible then $yz(e) > w_i(e) - \delta_i$.  However, by Lemma~\ref{lem:parity} (parity of $y$-values) we know $(yz(e) - w_i(e))/\delta_i$ is an integer,
so $yz(e)\ge w_i(e)$ before adjustment and $yz(e)\ge w_i(e) - \delta_i$ afterward (which could occur if both $u,v\in \hat{V}_{out}$), 
thereby preserving Property~\ref{prop:yz}(\ref{item:yz-lb}).
If $e$ is eligible then at least one of $u,v$ is in $\hat{V}_{in}$,
otherwise another blossom or augmenting path would have been formed, so $yz(e)$ cannot be reduced, 
which also preserves Property~\ref{prop:yz}(\ref{item:yz-lb}).

\paragraph{Case 2: $e \in M, \, u,v \in \hat{V}_{in}\cup \hat{V}_{out}$}
Since $u,v \in \hat{V}_{in}\cup \hat{V}_{out}$, Lemma~\ref{lem:parity} (parity of $y$-values) guarantees
that $(yz(e) - w_i(e))/\delta_i$ is an integer.  The only way $e$ can be ineligible is if $yz(e) = w_i(e) - \delta_i$ and $u,v\in \hat{V}_{in}$, 
hence $yz(e) = w_i(e)$ after dual adjustment, which preserves Property~\ref{prop:yz}(\ref{item:yz-lb},\ref{item:yz-ub}).
On the other hand, if $e$ is eligible then $u\in \hat{V}_{in}$ and $v\in \hat{V}_{out}$.  It cannot be that $u,v\in \hat{V}_{out}$, otherwise $e$ would
have been included in an augmenting path or root blossom.  In this case $yz(e)$ is unchanged, 
preserving Property~\ref{prop:yz}(\ref{item:yz-lb},\ref{item:yz-ub}).

\paragraph{Case 3: $e \not\in M, \, v \not\in \hat{V}_{in}\cup \hat{V}_{out}$}
If $e$ is eligible then $u\in \hat{V}_{in}$ and $yz(e)$ will increase.
If it is ineligible then $yz(e) \ge w_i(e) - \delta_i/2$ before adjustment
and $yz(e) \ge w_i(e) - \delta_i$ afterward.  In both cases Property~\ref{prop:yz}(\ref{item:yz-lb}) is preserved.

\paragraph{Case 4: $e \in M, \, v \not\in \hat{V}_{in}\cup \hat{V}_{out}$}
It must be that $e$ is ineligible, 
so $u\in \hat{V}_{in}$ and $yz(e) - w_i(e)$ is either negative or an odd multiple of $\delta_i/2$.
If $e$ is type $j$ then, by Property~\ref{prop:yz}(\ref{item:nonneg},\ref{item:yz-ub}) (granularity and near tightness), 
$yz(e) \le w_i(e) + 2(\delta_j - \delta_i) - \delta_i/2$ before adjustment
and $yz(e) \le w_i(e) + 2(\delta_j - \delta_i)$ afterward, preserving Property~\ref{prop:yz}(\ref{item:yz-ub}).
\end{proof}

Recall that Lemma~\ref{lem:approx} stated that the final matching will be a $(1-O(\epsilon))$-\MWM{}
if $\delta_L = O(\epsilon)$, free vertices have zero $y$-values, and $yz(e)-w(e) = O(\epsilon)\cdot w(e)$.
Lemmas~\ref{lem:weight-in-scale} and \ref{lem:approx-guarantee} establish these bounds.

\begin{lemma}\label{lem:weight-in-scale}
Let $i\le L$ be the scale index.  Then
\begin{enumerate}
\item For $i<L$, all edges eligible at any time in scales 0 through $i$ have weight at least $N/2^{i+1}+\delta_i$.\label{item:weight-in-scale1}
\item For any $i$, if $e\in M$ then $yz(e) \le (1+4\epsilon') w(e)$.\label{item:weight-in-scale2}
\end{enumerate}
\end{lemma}

\begin{proof}
We prove the parts separately.

\paragraph{Part 1} 
The last search for augmenting paths in scale $i$ begins when the $y$-values
of free vertices are $N/2^{i+2}$, and strictly less than $y$-values of other vertices, by Property~\ref{prop:yz}(\ref{item:y}).
An unmatched edge $e=(u,v)$ can only be eligible at this scale if
$yz(e) = w_i(e) - \delta_i \le w(e) - \delta_i$.  
Hence $w(e) \ge y(u) + y(v) + \delta_i \ge N/2^{i+1} + \delta_i$.

\paragraph{Part 2}
Let $e$ be a type $j$ edge in $M$ during scale $i$.
Property~\ref{prop:yz}(\ref{item:yz-ub}) states that
$yz(e) - w_i(e) \leq 2(\delta_j - \delta_i)$.  Since $w_i(e) \leq w(e)$ 
it also follows that
$yz(e) - w(e) \leq 2\delta_j - 2\delta_i < 2\delta_j = \epsilon' N/2^{j-1}$. 
By part 1, a type $j$ edge must have weight at least $N/2^{j+1} + \delta_j$,
hence $yz(e) - w(e) < 4\epsilon' \cdot w(e)$.
\end{proof}

\begin{lemma}\label{lem:approx-guarantee}
After scale $L = \log N$, $M$ is a $(1-5\epsilon')$-\MWM.
\end{lemma}

\begin{proof}
The final scale ends when free vertices have zero $y$-values.  Property~\ref{prop:yz}(\ref{item:yz-lb}) holds
with respect to~$\delta_L = \delta_0/2^L = \epsilon' N/2^L = \epsilon'$ and Lemma~\ref{lem:weight-in-scale}
states that $yz(e) \leq (1+4\epsilon') w(e)$.
By Lemma~\ref{lem:approx}, $w(M) \ge (1-5\epsilon')w(M^*)$.
\end{proof}

\begin{theorem}\label{thm:main1}
A $(1-\epsilon)$-\MWM{} can be computed in time $O(m\epsilon^{-1} \log N)$.
\end{theorem}

\begin{proof}
Each Augmentation and Blossom Shrinking step takes $O(m)$ time using a modified depth-first search~\cite[\S 8]{GT91}.  (Finding a maximal set of augmenting paths is significantly simpler, conceptually, than finding a maximal set of minimum-length augmenting paths, as is done in~\cite{MV80,Vazirani94}.)  Each Dual Adjustment step clearly takes linear time, as does the dual adjustment at the end of each scale.
Scale $i < L = \log N$ begins with free vertices' $y$-values at $N/2^{i+1} - \delta_i$ and ends with them at $N/2^{i+2} - \delta_i$.  Since $y$-values are decremented by $\delta_i/2$ in each Dual Adjustment step there are exactly
$(N/2^{i+2})/(\delta_i/2) = N/(2\delta_0) = \epsilon'^{-1}/2$ such steps.  
The final scale begins with free vertices' $y$-values at $N/2^{L+1} - \delta_{L}$ and ends with them at zero, so there are fewer than 
$(N/2^{L+1})/(\delta_L/2) = \epsilon'^{-1}$ Dual Adjustment steps.
Lemma~\ref{lem:approx-guarantee} guarantees that the final matching is a $(1-\epsilon)$-\MWM{} for
$\epsilon' \le \epsilon/5$.  Hence, the total running time is $O(m\epsilon^{-1}\log N)$.
\end{proof}

\subsection{A Linear Time Algorithm}

Our $O(m\epsilon^{-1}\log N)$-time algorithm requires few modifications to run in linear time, 
independent of $N$.  In fact, the algorithm as it appears in Figure~\ref{fig:alg} requires no modifications
at all: we only need to change the definition of {\em eligibility} and, in each scale, refrain from scanning edges
that cannot possibly be eligible or part of augmenting paths or blossoms.  
In light of Lemma~\ref{lem:weight-in-scale}(\ref{item:weight-in-scale1}) it is helpful to index edges
according to the first scale in which they may be eligible.

\begin{definition}
Define $\mu_i = N/2^{i+1} + \delta_i$, for $i<L$, and $\mu_L = 0$.
Define $\scale(e)$ to be the $i$ such that $w(e) \in [\mu_i,\mu_{i-1})$.
\end{definition}

Definition~\ref{def:eligible-linear} redefines eligibility.  The differences with
Definition~\ref{def:eligible} are underlined.

\begin{definition}\label{def:eligible-linear}
At scale $i$, an edge $e$ is {\em eligible} if at least one of the following hold:
\begin{enumerate}
\item $e\in E_B$ for some $B\in \Omega$.
\item $e\not\in M$ and $yz(e) = w_i(e) - \delta_i$.
\item $e\in M$ and $yz(e) - w_i(e)$  is a nonnegative intelligent multiple of $\delta_i$.
Furthermore, \underline{$\scale(e) \ge i - \gamma$,} \underline{where $\gamma \bydef \log\epsilon'^{-1}$}.
\end{enumerate}
Let $E_{elig}$ be the set of eligible edges 
and let $\Gelig = (V,E_{elig})/\Omega$ be the unweighted graph obtained by discarding ineligible edges and contracting
root blossoms.
\end{definition}

\begin{lemma}\label{lem:neartight}
Using Definition~\ref{def:eligible-linear} of eligibility rather than Definition~\ref{def:eligible}, Property~\ref{prop:yz}(\ref{item:nonneg},\ref{item:z},\ref{item:yz-lb},\ref{item:y}) is maintained and Property~\ref{prop:yz}(\ref{item:yz-ub}) (near tightness) holds in the following weaker form.
Let $e\in M\cup\bigcup_{B\in \Omega} E_B$ be a type $j$ edge with $\scale(e)=i$.  Then $yz(e) \le w_k(e) + 2(\delta_j - \delta_k)$ at any scale $k \in [i, i+\gamma]$
and $yz(e) \le w_k(e) + 3\delta_i < (1+6\epsilon')w(e)$ for $k > i+\gamma$.
\end{lemma}

\begin{proof}
In scales $i$ through $i+\gamma$ Property~\ref{prop:yz}(\ref{item:yz-ub}) is maintained as the two definitions of eligibility are the same.
At the beginning of scale $t = i+\gamma+1$, $e$ is no longer eligible and the $y$-values of free vertices are $N/2^{t+1} - \delta_t/2$.
From this moment on, the $y$-values of free vertices are incremented by a total of $\sum_{l \ge t+1} \delta_l$ (the dual adjustments following scales 
$t$ through $L -1$) and decremented a total of $N/2^{t+1} - \delta_t/2 + \sum_{l \ge t+1} \delta_l$ (in the Dual Adjustment steps
following searches for augmenting paths and blossoms).  Each adjustment to $y$-values by some quantity $\Delta$ may cause $yz(e)$ to increase by $2\Delta$.  
This clearly occurs in the dual adjustments following each scale as $y(u)$ and $y(v)$ are incremented by $\Delta$.  Following a search for blossoms it may be that $u,v \in \hat{V}_{in}$, which would also cause $y(u)$ and $y(v)$ to each be incremented by $\Delta$.  Note that $y(u),y(v)$ cannot be decremented in scales $t$ through $L$; if either were in $\hat{V}_{out}$ after a search for blossoms then $e$ would have been eligible, a contradiction.  Hence Property~\ref{prop:yz}(\ref{item:yz-lb}) (near domination) is maintained for $e$. 
Putting this all together, it follows that at any scale $k \geq t = i+\gamma+1 = i+\log\epsilon'^{-1}+1$,
\begin{align*}
yz(e) &\le w_k(e) + 2(\delta_j - \delta_k) + 2\cdot\paren{N/2^{t+1} - \delta_{t}/2 + 2\cdot\sum_{l \ge t+1} \delta_l}\hcm[-3]\\
	&< w_k(e) + 2\delta_i + 2\paren{\delta_{i+2} + \fr{3}{2}\delta_{t}} & \mbox{$j\ge i$, defn.~of $t$}\\
	&= w_k(e) + (2 + 1/2 + 3\epsilon'/2)\delta_i      & \mbox{defn.~of $\delta_t = \epsilon'\delta_{i}/2$}\\
	&< w_k(e) + 3\delta_i					& \mbox{$\epsilon' < 1/3$}\\
	&< (1+6\epsilon')w(e)				& \mbox{$w(e) \ge w_k(e) > N/2^{i+1} = (2\epsilon')^{-1}\cdot \delta_i$}
\end{align*}
which proves the claim.  Note that the inequality $yz(e) < w_k(e) + 3\delta_i$ also holds in scales
$k\in [i,i+\gamma]$.  If $e$ is type $j\ge i$ then $yz(e) \le w_k(e) + 2(\delta_j - \delta_k) < w_k(e) + 2\delta_i$.
This fact will be used in Lemma~\ref{lem:incident}.
\end{proof}

The algorithm will deliberately ignore unmatched edges that may still be eligible according to 
Definition~\ref{def:eligible-linear}.  This will be justified on the grounds that 
such edges {\em must} be adjacent to matched ineligible edges, and therefore cannot be contained in an eligible augmenting path.
Lemma~\ref{lem:incident} will be used to argue that scale-$i$ edges can be safely ignored after scale $i+\gamma+2$.

\begin{lemma}\label{lem:incident}
Let $e_1= (u,v)$ be an edge with $\scale(e_1) = i$ 
and let $e_0 = (u',u)$ and $e_2 = (v,v')$ be the $M$-edges incident to $u$ and $v$ at
some time after scale $i$.  Then at least one of $e_0$ and $e_2$ exists, say $e_0$,
and $\scale(e_0) \le i+2$.
\end{lemma}

\begin{proof}
Following the last Dual Adjustment step in scale $i$ the $y$-values of free vertices are $N/2^{i+2} - \delta_i/2$.
It cannot be that both $u$ and $v$ are free at this time, otherwise $yz(e_1) = y(u) + y(v) = N/2^{i+1} - \delta_i = \mu_i - 2\delta_i < w_i(e_1) - \delta_i$, violating Property~\ref{prop:yz}(\ref{item:yz-lb}) (near domination).  Hence, either $u$ or $v$ is matched for the remainder of the computation.  If $e_1$ is matched the claim is trivial, so, assuming the claim is false, whenever $e_0,e_2$ exist we have $\scale(e_0), \scale(e_2) \ge i+3$.  That is, $w(e_0),w(e_2) < \mu_{i+2} = N/2^{i+3} + \delta_{i+2}$.

It cannot be that
$e_1$ is in a blossom without $e_0$ or $e_2$ also being in the blossom.  For any $l\in\{0,1,2\}$ let $\mathcal{B}_l \subseteq \Omega$ be the blossoms containing $e_l$ at a given time.  The laminarity of blossoms ensures that either $\mathcal{B}_1 \subseteq \mathcal{B}_0$ or $\mathcal{B}_1\subseteq \mathcal{B}_2$.  Suppose it is the former, that is, $e_0$ exists and $e_2$ may or may not exist.
Then, if the current scale is $k\ge i+3$, by Property~\ref{prop:yz}(\ref{item:yz-lb}) (near domination)
$yz(e_1) = y(u) + y(v) + \sum_{B\in \mathcal{B}_1}z(B) \ge w_k(e_1) - \delta_k$.
By Lemma~\ref{lem:neartight} 
$y(u) + \sum_{B\in\mathcal{B}_1} z(B) < yz(e_0) \le w_k(e_0) + 3\delta_{i+3}$
and, if $e_2$ exists, 
$y(v) < yz(e_2) \le w_k(e_2) + 3\delta_{i+3}$.  
These inequalities follow from the definition of $yz$, the containment 
$\mathcal{B}_1\subseteq \mathcal{B}_0$ and the fact that $e_0$ and $e_2$ can only be at scale $i+3$ or higher.
Without loss of generality 
we can assume $y(u) + \sum_{B\in\mathcal{B}_1} z(B) \ge y(v)$.
(If $e_2$ exists and $y(v) > y(u) + \sum_{B\in\mathcal{B}_1} z(B)$ then $e_2$ takes the role 
of $e_0$ below.)
Putting these inequalities together we have

\begin{align*}
w_k(e_1) &\le y(u) + y(v) + \sum_{B\in \mathcal{B}_1}z(B) + \delta_k & \mbox{near domination}\\
	&\le 2\paren{y(u) + \sum_{B\in \mathcal{B}_1} z(B)} + \delta_k  & \mbox{$y(u) + \sum_{B\in\mathcal{B}_1} z(B) \ge y(v)$}\\
	&\le 2\paren{yz(e_0)} + \delta_k					& \mbox{$\mathcal{B}_1 \subseteq \mathcal{B}_0$}\\
	&< 2(w_k(e_0) + 3\delta_{i+3}) + \delta_k  & \mbox{Lemma~\ref{lem:neartight}}\\
	&< 2w(e_0) + 7 \delta_{i+3}  & \mbox{$k\ge i+3$, $\epsilon' < 1/3$}\\
\intertext{and therefore}
w(e_0) &> \fr{1}{2}(w_k(e_1) - \delta_i)	& \mbox{$7\delta_{i+3} < \delta_i$}\\
           &\ge N/2^{i+2}	& \mbox{$\scale(e_1)=i, w_k(e_1) \ge \mu_i = N/2^{i+1}+\delta_{i}$}\\
           &> N/2^{i+3} + \delta_{i+2} = \mu_{i+2}
\end{align*}
This contradicts the fact that $\scale(e_0) \ge i+3$, since, by definition, such edges have
$w(e_0) < \mu_{i+2}$.
\end{proof}

\begin{theorem}\label{thm:main2}
A $(1-\epsilon)$-\MWM{} can be computed in time $O(m\epsilon^{-1} \log\epsilon^{-1})$.
\end{theorem}

\begin{proof}
We execute the algorithm from Figure~\ref{fig:alg} where $\Gelig$ refers to the eligible subgraph as
defined in Definition~\ref{def:eligible-linear}.
We need to prove several claims: (i) the algorithm does, in fact return a $(1-\epsilon)$-\MWM{} for suitably chosen
$\epsilon' = \Theta(\epsilon)$, (ii) the number of scales in which an edge could conceivably participate in an augmenting path or blossom is $\log\epsilon^{-1} + O(1)$, and (iii) it is possible in linear time to compute the scales in which each edge must participate.
Part (i) follows from Lemmas~\ref{lem:approx} and \ref{lem:neartight}.
Since $yz(e) \le (1+6\epsilon')w(e)$ for any $e\in M$ and $\delta_L= \epsilon'$,
Lemma~\ref{lem:approx} implies that $M$ is a $(1-\epsilon)$-\MWM{} when $\epsilon' \le \epsilon/7$.

Turning to part (ii), consider an edge $e$ with $\scale(e) = i$.  By Lemma~\ref{lem:weight-in-scale}(\ref{item:weight-in-scale1}) $e$ can be ignored in scales 0 through $i-1$.   If $e=(u,v) \in M$ then, according to Definition~\ref{def:eligible-linear}, $e$ will be ineligible in scales $i+\gamma+1$ through $\log N$.
After scale $i+\gamma$ no augmenting path or blossom can contain $e$, so we can commit it to the final matching and remove from consideration all edges incident to $u$ or $v$.  Now suppose that $e\not\in M$ at the end of scale $i+\gamma+2$.  Lemma~\ref{lem:incident} states that either $u$ or $v$ is incident to a matched edge $e_0$ with $\scale(e_0) \le i+2$, which by the argument above, will be committed to the final matching, thereby removing $e$ from further consideration.  Thus, to faithfully execute the algorithm we only need to consider $e$ in scales $\scale(e)$ through $\scale(e) + \gamma + 2$, that is, 
$\gamma+3 = \log\epsilon'^{-1} + 3 \le \log\epsilon^{-1} + 6$ scales in total.  

We have narrowed our problem to that of computing $\scale(e)$ for all $e$.  This is tantamount to computing the most significant bit ($\msb(x) = \floor{\log_2 x}$) in the binary representation of $w(e)$.  Once the $\msb$ is known, $\scale(e)$ can be just one of two possible values.  MSBs can be computed in a number of ways using standard instructions.  It is trivial to extract $\msb(x)$ after converting
$x$ to floating point representation.  Fredman and Willard~\citeyear{FW93} gave an $O(1)$ time algorithm using unit time multiplication.  However, we do not need to rely on floating point conversion or multiplication.
In Section~\ref{sect:def} we showed that without loss of generality $\log N \le 2\log n$.
Using a negligible $O(n^\beta)$ space and preprocessing time we can tabulate the answers on $\beta \cdot \log n$-bit integers, where $\beta \le 1$, then compute MSBs with $2\beta^{-1}=O(1)$ table lookups.
\end{proof}

\section{Exact Maximum Weight Matching}\label{sect:exact}

At a high level our exact \MWM{} algorithm is similar to our $(1-\epsilon)$-\MWM{} algorithm. 
It consists of $\log N + 1$ scales, where, in the $i$th scale, the magnitude of 
dual adjustments and the violation of domination/tightness is bounded in terms of $\delta_i$, which decreases geometrically with $i$.
However, beyond this similarity the two algorithms are quite different.
Our exact \MWM{} algorithm only works on bipartite graphs; we assume for simplicity that the graph consists of exactly $n$ {\em left}
vertices and $n$ {\em right} vertices.

We redefine $\delta_0 = 2^{\floor{\log(N/\sqrt{n})}}$ and let $\delta_i = \delta_0/2^i$ and $w_i(e) = \delta_i \floor{w(e)/\delta_i}$
be the granularity and weight function of the $i$th scale, where $i\in[0,L]$ and $L = \ceil{\log N}$.\footnote{Note that if $N\le \sqrt{n}$ we do not require a different weight function at each scale since $w_i = w$ for all $i\ge 0$.}
The algorithm maintains a matching $M$ and duals $y$ satisfying Property~\ref{invariant}.
(As the graph is bipartite there is no need for blossoms or their duals $z$.)
Whereas Property~\ref{prop:yz} allows domination to be violated by $\delta_i$ but enforces tightness of matched edges (of type $i$),
Property~\ref{invariant} enforces domination but lets tightness be violated by up to $3\delta_i$.

\begin{property}\label{invariant} 
Let $i \in [0,L]$ be the scale, $M$ be the current matching, and $y \::\: V\rightarrow \mathbb{R}_{\ge 0}$ be the vertex duals,
where $y(e) = y(u) + y(v)$ for edge $e=(u,v)$.
\begin{enumerate}
\item \emph{Granularity:} $y(u)$ is a nonnegative multiple of $\delta_i$.\label{granularity}
\item \emph{Domination:} $y(e) \geq w_{i}(e)$ for all $e \in E$.\label{domination}
\item \emph{Near Tightness:} Let $e\in M$ be a matched edge.
In scale 0, $y(e) \le w_0(e)+\delta_0$.  Throughout scale $i$ we have
$y(e) \leq w_{i}(e) + 3 \delta_{i}$ and at the end of scale $i$ we have
$y(e) \leq w_{i}(e) + \delta_{i}$.
\label{tightness}
\item \emph{Free Vertex Duals:} In scale $i=0$, right free vertices have zero $y$-values and left free vertices have equal and minimal $y$-values among left vertices.
At the end of scale 0 and throughout scales $1,\ldots,L$, all free vertices have zero $y$-values.\label{freeduals}
\end{enumerate}
\end{property}

\begin{lemma}\label{firstlemma}
Let $M$ be the matching at the end of scale $i$ and $M^*$ be the \MWM.  
Then $w(M) \ge w(M^*) - 2n\delta_i$, and when $i=L$, $w(M) \ge w(M^*) - n\delta_L \ge w(M^*) - \sqrt{n}$.
\end{lemma}

\begin{proof}
By Property~\ref{invariant} we have
\begin{align*}
w(M) \ge w_i(M)  &\ge \sum_{e\in M} y(e) - n\delta_i	& \mbox{near tightness}\\
			&= \sum_{u\in V} y(u) - n\delta_i		& \mbox{free vertex duals}\\
			&\ge w_i(M^*) - n\delta_i			& \mbox{domination, non-negativity of $y$}\\
			&\ge w(M^*) - 2n\delta_i				& \mbox{defn. of $w_i$, $|M^*|\le n$}
\end{align*}
Note that the last inequality is weak whenever $\delta_i \leq 1$, since in this case $w = w_i$.
Specifically, when $i=L$ we have $\delta_L = 2^{\floor{\log(N/\sqrt{n})} - \ceil{\log N}} \le 1/\sqrt{n}$,
implying $w(M) \ge w(M^*) - n\delta_L \ge w(M^*) - \sqrt{n}$.
\end{proof}

As in our $(1-\epsilon)$-\MWM{} algorithm we restrict our attention to augmentations on eligible edges.  However, our definition of {\em eligibility}
depends on the context.  For integers $0\le a \le b$, the eligibility graph $G[a,b]$ at scale $i$ consists of all edges $e$ such that
\begin{itemize}
\item[--- ] $e\not\in M$ and $y(e) = w_i(e)$, or
\item[--- ] $e\in M$ and $w_i(e) + a \delta_i \le y(e) \le w_i(e) + b\delta_i$.
\end{itemize}

The algorithm consists of three phases, each with a distinct goal.  Phase I, Phase III, and each scale of Phase II will require $O(m\sqrt{n})$ time, for a total of $O(m\sqrt{n}\log N)$ time.

\begin{description}
\item[Phase I] The phase operates only at scale 0.  It is a simplified execution of the Gabow-Tarjan~\citeyear{GT89} algorithm, stopping not when $M$ is perfect but when $y$-values of free vertices are zero.  In this phase an augmentation is an augmenting path in $G[1,1]$ whose endpoints are free.\\

\item[Phase II] The phase operates at scales $i=1,\ldots,L$.  At the beginning of the scale $M \subseteq G[0,3]$.  The goal is to eliminate $M$-edges that violate near-tightness
by $2\delta_i$ or $3\delta_i$, that is, the scale ends when $M \subseteq G[0,1]$.  In this phase an augmentation is either an augmenting cycle in $G[1,3]$ or an augmenting path in $G[1,3]$ whose endpoints have zero $y$-values.  Note that the ends of augmenting paths can be either free vertices or matched  edges.\\

\item[Phase III] The phase operates only at scale $L$.  The last scale of Phase II leaves $M \subseteq G[0,1]$.  In this phase an augmentation is either an augmenting cycle or
an augmenting path whose ends have zero $y$-values, that, in addition, contains at least one non-tight edge.  That is, the augmenting path/cycle must exist in $G[0,1]$ but not $G[0,0]$,
which implies that augmenting along such a path increases the weight of the matching.
By Lemma~\ref{firstlemma}, at the end of Phase II  $w(M) \ge w(M^*) - \sqrt{n}$.  
We guarantee that $w(M)$ is a nondecreasing function of time,
so, by the integrality of edge weights, $w(M)$ can be improved at most $\sqrt{n}$ times in Phase III.
\end{description}

The following notation will be used liberally.
Let $X\subseteq V(G)$ be a vertex set, $H$ be a subgraph of $G$, and $M$ be an arbitrary matching.
Define $\VODD(X,H)$ (and $\VEVEN(X,H)$) to be the set of vertices reachable from $X$ in $H$ by an odd-length (and even-length) alternating 
path starting with an \underline{\em unmatched} edge.  The directed graph $\vec{H}$ is obtained by
orienting $e\in E(H)$ from left to right if $e\not\in M$ and from right to left if $e\in M$.  In our algorithm $H$ is always chosen to be $G[a,b]$ for some parameters $a,b$.
Note that $\VODD(X,H), \VEVEN(X,H)$, and $\vec{H}$ are defined with respect to a matching known from context.  It is clear that
$\VODD(X,H)$ and $\VEVEN(X,H)$ can be computed in linear time, for example, with depth first search (DFS).

\subsection{Phase I}

In Phase I we operate on $G[1,1]$.
We begin with an empty matching $M = \emptyset$ and let $y(v) = 0$ for right vertices and $y(v) = \delta_0\floor{N/\delta_0}$ for left vertices.
This clearly satisfies Property~\ref{invariant}(\ref{domination}) (domination) since $w_0(e) = \delta_0\floor{w(e)/\delta_0}$ and the maximum edge weight is $N$.
The Phase I algorithm oscillates between augmenting along a maximal set of eligible augmenting paths
and performing dual adjustments.  Note that if we augment along an eligible augmenting path, 
all edges of the path become ineligible.  See Algorithm~\ref{alg1} for the details.

	\begin{algorithm}
		\caption{Phase I}
		\label{alg1}
		\begin{algorithmic}
			\STATE {\bf Initialization:}
			\STATE $M \leftarrow \emptyset$
			\STATE $y(v) \leftarrow \begin{cases} \delta_0\floor{N/\delta_0} & \mbox{if $v$ is a left vertex}  \\  0  & \mbox{if $v$ is a right vertex} \end{cases}$
			\REPEAT
			\STATE {\bf Augmentation:}
			\STATE Find a maximal set $\Psi$ of augmenting paths in $G[1,1]$ and set $M \leftarrow M \oplus \bigcup_{P\in \Psi} P$.
			\STATE
			\STATE {\bf Dual Adjustment:}
			\STATE Let $F$ be the set of left free vertices
			\STATE $y(v) \leftarrow \begin{cases}
					y(v) - \delta_0 & \mbox{if $v \in \VEVEN(F, G[1,1])$}\\
					y(v) + \delta_0 & \mbox{if $v \in \VODD(F, G[1,1])$}\\
					y(v) & \mbox{otherwise}
					\end{cases}$
			\UNTIL {$y$-values of left free vertices are zero or all left vertices are matched}
		\end{algorithmic}
	\end{algorithm}

After an augmentation step, there cannot be any augmenting paths in $G[1,1]$, which implies no free vertex is in $\VODD(F, G[1,1])$. 
Property~\ref{invariant}(\ref{freeduals}) is maintained for left free vertices since their $y$-values are reduced in lockstep in every dual adjustment step.
It is also preserved for right free vertices since they are not in $\VODD(F, G[1,1])$ and therefore never have their $y$-values adjusted.
Property~\ref{invariant}(\ref{domination}) (domination) is maintained since no eligible edge 
can have one endpoint in $\VEVEN(F,G[1,1])$ without the other being in $\VODD(F,G[1,1])$.
Property~\ref{invariant}(\ref{tightness}) (near tightness) is maintained since for any $e\in M$,
$y(e)$ is unchanged if $e$ is eligible, and, if $e$ is ineligible (that is, $y(e) = w_0(e)$), $y(e)$ may only be incremented by $\delta_0$.

The number of augmentation/dual adjustment steps is bounded by the number of dual adjustments, 
that is, $(\delta_0\floor{N/\delta_0}) / \delta_0 \le N/2^{\floor{\log(N/\sqrt{n})}} < 2\sqrt{n}$. 
Thus, Phase I takes $O(m \sqrt{n})$ time.

\subsection{Phase II}\label{sect:phaseII}
At the beginning of scale $i \in [1,L]$ in Phase II we set $y(u) \leftarrow y(u) + \delta_{i}$ for each left vertex $u$ and leave the $y$-values of right vertices unchanged.
Since $w_i(e) \le w_{i-1}(e) + \delta_i$, this preserves Property \ref{invariant}(\ref{domination}) (domination).  
Property \ref{invariant}(3) (near tightness) is also maintained.
\begin{align}
y(e) & \leftarrow y(e) + \delta_{i} \nonumber \\
        & \leq w_{i-1}(e) + \delta_{i-1} + \delta_{i}   &&\mbox{by Property \ref{invariant}(3) at the end of scale $i-1$} \nonumber \\
        & \leq w_{i}(e) + 3\delta_{i} && \mbox{since $\delta_{i-1} = 2\delta_{i}$  and $w_{i-1}(e) \leq w_i(e)$} \nonumber
\end{align}
However, Property \ref{invariant}(\ref{freeduals}) may be violated since the $y$-values of left free vertices are $\delta_i$, not zero. 
Hence, we will run one iteration of Phase I's augmentation and dual adjustment steps on $G[1,3]$.  These steps preserve domination and near tightness and
bring left free vertices' $y$-values down to zero, restoring Property~\ref{invariant}(\ref{freeduals}).  
This procedure takes $O(m)$ time but is executed just once for each of the $\log N$ scales.

Next, we will repeatedly perform Phase II Augmentation and Phase II Dual Adjustment steps (see Sections \ref{sec:aug} and \ref{sec:dual}) on $G[1,3]$ until 
$M \cap G[2,3] = \emptyset$, that is $y(e) \leq  w_{i}(e) + \delta_{i}$ for all $e \in M$.  At this point scale $i$ ends and scale $i+1$ begins.

\subsection{Phase II Augmentation} \label{sec:aug}

The goal of an augmentation step in Phase II is simply to eliminate all augmenting paths and cycles from $G[1,3]$.
We do this in two steps, first eliminating augmenting cycles then paths.  Notice that in contrast to Phase I, 
augmenting paths may start or end with matched edges.

In the first stage of augmentation we will find a maximal set of vertex-disjoint augmenting cycles $\mathcal{C}$ using DFS.  
Observe that directed cycles in $\vec{G}[1,3]$ correspond to augmenting cycles in $G[1,3]$; this simplifies the DFS algorithm since we do not need to distinguish matched and unmatched edges.
At all times the DFS stack $\mathcal{S}$ forms an alternating path.  If there is a back edge from the top of the stack to another vertex on this stack,
that is, $\mathcal{S} = (\ldots,v,\ldots,u)$ and $(u,v)\in E(\vec{G}[1,3])$, this back edge closes an augmenting cycle, which can be added to $\mathcal{C}$.
A vertex is {\em marked} when it is popped off the DFS stack, either because the vertex joins an augmenting cycle in $\mathcal{C}$ or if it is found not to be contained
in any augmenting cycle.  Algorithm~\ref{alg2} gives the details for {\em Cycle-Search}.

	\begin{algorithm}
		\caption{{\em Cycle-Search}: returns a maximal set of eligible augmenting cycles}
		\label{alg2}
		\begin{algorithmic}
			\STATE {All vertices are initially {\em unmarked}}
			\STATE {$\mathcal{C} \leftarrow \emptyset$ \hfill \{set of augmenting cycles discovered so far\}\ \ \ \ \ \ \ }
			\STATE {$\mathcal{S} \leftarrow ()$ \hfill \{DFS stack\}\ \ \ \ \ \ \ }
			\WHILE {there are still unmarked vertices}
			\IF {$\mathcal{S}$ is empty}
			\STATE {Push any unmarked $u_0\in V$ onto $\mathcal{S}$.}
			\ENDIF
			\STATE {$u$ = the top of $\mathcal{S}$}
			\STATE {$V_u = \{v \:|\: \mbox{$(u,v)\in E(\vec{G}[1,3])$  and  $v$ is unmarked}\}$}
			\IF {$V_u$ is empty}
			\STATE {Mark $u$ and pop $u$ off $\mathcal{S}$}
			\ELSE
				\STATE {Choose any $v\in V_u$}
				\IF {$v$ appears in $\mathcal{S}$, that is, if $\mathcal{S} = (\ldots,v,\ldots,u)$}
					\STATE {$\mathcal{C} \leftarrow \mathcal{C} \cup \{(v,\ldots,u,v)\}$ \hfill \{add augmenting cycle to $\mathcal{C}$\}\ \ \ \ \ \ \ }
					\STATE {Mark $v,\ldots,u$ and pop $(v,\ldots,u)$ off $\mathcal{S}$}
				\ELSE
					\STATE {Push $v$ on $\mathcal{S}$}
				\ENDIF
			\ENDIF
			\ENDWHILE
			\RETURN  $\mathcal{C}$
		\end{algorithmic}
	\end{algorithm}

\begin{lemma}\label{nocycle} The algorithm {\em Cycle-Search} finds a maximal set of vertex-disjoint augmenting cycles $\mathcal{C}$. Moreover, if we augment along every cycle in $\mathcal{C}$, then the graph $G[1,3]$ contains no more augmenting cycles. \end{lemma} 
\begin{proof}
Supposing $\mathcal{C}$ is not maximal, let $C = (v_0, v_1, \dots v_{k-1}, v_0)$ be any cycle vertex-disjoint from all cycles in $\mathcal{C}$, where $v_0$ is the first vertex of $C$ pushed onto the stack. Let $t$ be the largest index such that $v_t$ is pushed onto the stack before $v_0$ is popped off.  
It follows that $v_{t+1 \operatorname{mod} k}$ is unmarked and therefore appears in $V_{v_t}$.  If $t=k-1$ then $v_0\in V_{v_{k-1}}$ and the search will discover an augmenting cycle containing $v_t$; if $t<k-1$ then the search will push $v_{t+1}$ onto the stack, contradicting the maximality of $t$.

If there exists an eligible cycle $C$ after augmentation, then this cycle must share a vertex $v$ with some cycle $C' \in \mathcal{C}$ due to the maximality of $\mathcal{C}$.
However, since $C'$ contains $v$'s mate (both before and after augmentation), $C$ and $C'$ must intersect at an edge, which contradicts the fact that all edges in $C'$ become
ineligible after augmentation.
\end{proof}

\begin{figure}[h]
\begin{center}
\scalebox{.45}{\includegraphics{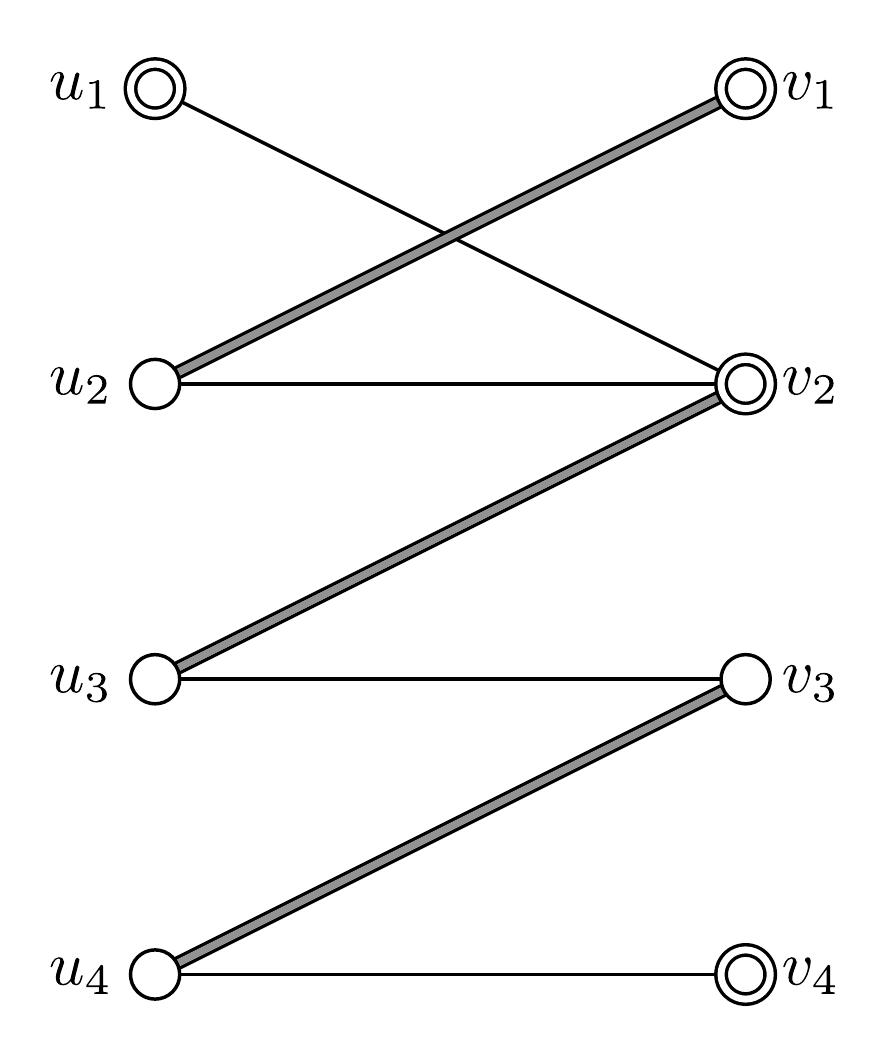}}
\end{center}
\caption{\label{fig:maximal_aug_path} 
An illustration of starting vertices and maximal augmenting paths in $G[1,3]$. 
The plain edges denote unmatched edges; the shaded edges are matched. 
The haloed vertices have zero $y$-values. 
The set of starting vertices is $\{u_1, v_1,v_2\}$. 
The path $P=(v_2,u_3,v_3,u_4,v_4)$ is an augmenting path
but not a maximal augmenting path, since it can be extended to a longer one.
For example, $(u_1,v_2,u_3,v_3,u_4,v_4)$ and $(v_1,u_2,v_2,u_3,v_3,u_4,v_4)$
are augmenting paths containing $P$.}
\end{figure}

In the second stage of augmentation 
we will eliminate all the augmenting paths in $G[1,3]$. This is done by finding a maximal set of vertex-disjoint {\it maximal} augmenting paths,
which are those not properly contained in another augmenting path.  (Recall that in Phase II, augmenting paths can end at matched edges so long as the endpoints of the path
have zero $y$-values.  Any augmenting path with free endpoints is necessarily maximal, but one ending in matched edges may not be maximal.)  
See Figure \ref{fig:maximal_aug_path} for an illustration. 

Consider the graph $\vec{G}[1,3]$. It must be a directed acyclic graph, since, by Lemma~\ref{nocycle}, $G[1,3]$ does not contain augmenting cycles.
Call a vertex with zero $y$-value a {\em starting} vertex if it is left and free or right and matched, and an {\em ending} vertex if it is left and matched or right and free.
Let $A$ and $B$ be the set of starting and ending vertices.  It follows that any augmenting path in $G[1,3]$ corresponds to a directed path from an $A$-vertex to a $B$-vertex
in $\vec{G}[1,3]$.  With this observation in hand we can find a maximal set $\mathcal{P}$ of maximal augmenting paths using DFS.  We initiate the search on each $A$-vertex 
$u$ in topological order.  When an augmenting path to a $B$-vertex, say $v$, is first discovered we cannot commit $(u,\ldots,v)$ to $\mathcal{P}$ immediately.  
Rather, we continue to look for an even longer augmenting path, and keep $(u,\ldots,v)$ only after all outgoing edges from $v$ have been exhausted.
Algorithm~\ref{alg3} gives the {\em Path-Search} procedure.

	\begin{algorithm}
		\caption{{\em Path-Search}: returns a maximal set of maximal eligible augmenting paths}
		\label{alg3}
		\begin{algorithmic}
			\STATE {All vertices are initially {\em unmarked}}
			\STATE {$A \leftarrow \{u \:|\: \mbox{$y(u)=0$ and $u$ is left and free or right and matched}\}$ \hfill \{starting vertices\}\ \ \ \ \ \ \ }
			\STATE {$B \leftarrow \{u \:|\: \mbox{$y(u)=0$ and $u$ is right and free or left and matched}\}$ \hfill \{ending vertices\}\ \ \ \ \ \ \ }
			\STATE {$\mathcal{P} \leftarrow \emptyset$ \hfill \{set of augmenting cycles discovered so far\}\ \ \ \ \ \ \ }
			\STATE {$\mathcal{S} \leftarrow ()$ \hfill \{DFS stack\}\ \ \ \ \ \ \ }
			\WHILE {there are still unmarked $A$-vertices}
			\IF {$\mathcal{S}$ is empty}
				\STATE {Push the first (topologically) unmarked $u_0 \in A$ onto $\mathcal{S}$.}
			\ENDIF
			\STATE {$u$ = the top of $\mathcal{S}$}
			\STATE {$V_u = \{v \:|\: \mbox{$(u,v)\in E(\vec{G}[1,3])$  and  $v$ is unmarked}\}$}
			\IF {$V_u$ is empty}
				\IF {$u\in B$} 
					\STATE {$\mathcal{P} \leftarrow \mathcal{P} \cup \{\mathcal{S}\}$ \hfill \{$\mathcal{S}$ forms a maximal augmenting path\}\ \ \ \ \ \ \ }
					\STATE {Mark all vertices in $\mathcal{S}$ and set $\mathcal{S} \leftarrow ()$ \hfill \{Pop all vertices off the stack\}\ \ \ \ \ \ }
				\ELSE
					\STATE {Mark $u$ and pop $u$ off $\mathcal{S}$}
				\ENDIF
			\ELSE
				\STATE {Push an arbitrary $v\in V_u$ onto $\mathcal{S}$}
			\ENDIF
			\ENDWHILE
			\RETURN  $\mathcal{P}$
		\end{algorithmic}
	\end{algorithm}

\begin{lemma}\label{nopath} 
After augmenting along every path in $\mathcal{P}$, the graph $G[1,3]$ contains no augmenting paths. 
\end{lemma}

\begin{proof}
Suppose that there exists an augmenting path $Q$ after the augmentation. Then, by the maximality of $\mathcal{P}$, there must be some augmenting path in $\mathcal{P}$ sharing vertices with $Q$. There can be two cases.
\begin{description}
\item[Case 1] There exists a $P\in\mathcal{P}$ and a $v \in P \cap Q$ that is not an endpoint of $P$.
Since $P$ contains $v$ and its mate before and after augmentation, $P$ and $Q$ must share an edge, which is impossible since all edges in $P$ become ineligible after augmentation.

\item[Case 2] Let $P$ be the first path added to $\mathcal{P}$ that intersects $Q$.  Since we are not in Case 1, $P$ and $Q$ intersect at a common endpoint, say $x$.
(Note that the only way this is possible is if $x$ is matched before augmentation
and free afterward.  If $x$ were a free endpoint of $P$ before augmentation it would be incident to a matched ineligible edge after augmentation and could therefore 
not be an endpoint of $Q$.)
Let $x_P$ and $x_Q$ be the other endpoints of $P$ and $Q$. If $x_P = x_Q$ then $G[1,3]$ contained an augmenting cycle, contradicting Lemma~\ref{nocycle}. 
If $x_P\in A$ is a starting vertex, consider the moment when the stack contained only $x_P$.  At this time $PQ$ is an augmenting path of unmarked vertices,
so {\em Path-Search} would not place the non-maximal augmentation $P$ in $\mathcal{P}$.
On the other hand, if $x$ is a starting vertex, it must be a right matched vertex that becomes free after augmentation, which implies that $x_Q$ must also be a starting vertex. 
Since our search explores $A$-vertices in topological order, the search from $x_Q$ would have preceded the search from $x$.
Thus, the first augmentation of $\mathcal{P}$ intersecting $Q$ must contain $x_Q$, contradicting the fact that $P$ does not contain $x_Q$.
\end{description}
\end{proof}

\subsection{Phase II Dual Adjustment} \label{sec:dual}

Recall that scale $i$ ends when $M\subseteq G[0,1]$.
Let $B = M\cap G[2,3]$ be the {\em bad} edges and let $f : E\rightarrow \{0,1,2\}$ measure the {\em badness}, defined as follows.
\[
f(e) = \left\{
\begin{array}{ll}
\f{y(e) - w_i(e) - \delta_i}{\delta_i}	& \mbox{ for $e\in B$}\\
0							& \mbox{ for $e\not\in B$}
\end{array}
\right.
\]
Let $f(M) = \sum_{e\in B} f(e)$ be the total badness of $M$. 
The goal of dual adjustment is to eliminate $B$, or equivalently, to reduce $f(M)$ to 0.  We will show how to reduce $f(M)$ by roughly $\sqrt{f(M)}$ in linear time.

A $B' \subseteq B$ is called a {\em chain} if there is an alternating path in $G[1,3]$ containing $B'$ and 
an {\em anti-chain} if no alternating path in $G[1,3]$ contains two edges $e_1,e_2 \in B'$.
Lemma~\ref{chainanti} basically follows from Dilworth's lemma \citeyear{Dilworth50}. 

\begin{lemma}\label{chainanti}
For any $t > 1$, there exists $B' \subseteq B$ such that either $B'$ is a chain with $f(B') \geq  \ceil{t}$ or $B'$ is an anti-chain with $|B'| \geq \ceil{f(M) / 2t}$. Moreover, such a $B'$ can be found in linear time.
\end{lemma}

\begin{proof}
Let $S$ be the set of vertices with zero in-degree in the acyclic graph $\vec{G}[1,3]$.
In linear time we compute distances from $S$ using $-f$ as the length function. Let $d(v)$ be the distance to $v$.
Suppose there is a $v$ with $d(v) \le -\ceil{t}$ and let $P$ be a shortest path to $v$.  It follows that $B' = P\cap B$ is a chain with $f(B') \geq \ceil{t}$.
If this is not the case then, for every $(u,v) \in B$ (where $v$ is a left vertex), $d(v) \in [-(\ceil{t} - 1), -1]$. 
Since $f(e) \leq 2$ for $e \in B$, we must have at least $\ceil{|B| / (\ceil{t} - 1)} > \ceil{f(M) / 2t}$ such $v$ with a common distance, say $-k$.
It follows that $B' = \{ (u,v) \in B \:|\: \mbox{$v$ is a left vertex and $d(v) = -k$} \}$ is an anti-chain, for if $(u_1,v_1),(u_2,v_2)\in B'$ were on an alternating path in $G[1,3]$,
the distance to $v_2$ would be strictly smaller than the distance to $v_1$.
\end{proof}

Below we show that if $B'$ is a chain we can decrease the total badness by $f(B')$ in linear time. 
On the other hand, if $B'$ is an anti-chain, then we can decrease the total badness by $|B'| / 2$, also in linear time.

\subsubsection {Phase II Dual Adjustment: Antichain Case} 

When performing dual adjustments we must be careful to maintain Property~\ref{invariant}, 
which states that $y(u)$ must always be nonnegative, and must be zero if $u$ is free.  
This motivates the definition of a {\em dual adjustable} vertex.

\begin{definition}
A vertex $u$ is said to be {\em dual adjustable} if every vertex in $\VODD(u, G[1,3])$ is matched and every vertex $v \in \VEVEN(u, G[1,3])$ has $y(v) > 0$.
\end{definition}

\begin{lemma}\label{uvadjust} For every $e = (u,v) \in B$, 
either $u$ is adjustable or $v$ is adjustable or both. Furthermore, all adjustable vertices can be found in $O(m)$ time.
\end{lemma}

\begin{proof}
First, suppose that for some $e=(u,v) \in B$, both $u$ and $v$ are not adjustable.
There must exist vertices $w$ and $x$ having zero $y$-values where $(w,\ldots,u,v,\ldots,x)$ is an augmenting path in $G[1,3]$.
However, this contradicts Lemma~\ref{nopath}, which states that there are no augmentations in $G[1,3]$ after an augmentation step.
Let $\tilde{V} = \{v \:|\: \mbox{$v$ is free or $(v,v')\in M$ and $y(v')=0$} \}$.  By definition a vertex is {\em not} dual adjustable if and only if it lies in 
$\VODD(\tilde{V}, G[1,3])$, which can be computed in linear time.
\end{proof}

Let $B' \subseteq B$ be an anti-chain. 
The procedure {\em Antichain-Adjust}$(B')$ selects a set of dual adjustable vertices $X$ incident to $B'$ and on a common side (left or right), then does a dual adjustment starting at $X$. Since, by Lemma \ref{uvadjust}, for any $(u,v) \in B'$ either $u$ is adjustable or $v$ is adjustable or both, we can guarantee that $|X| \geq |B'| / 2$. 
See Figure \ref{fig:anti_chain} for an example.

	\begin{algorithm}
		\caption{{\em Antichain-Adjust}$(B')$}
		\label{alg4}
		\begin{algorithmic}
		\STATE {}
			\STATE {\bf Find adjustable vertices:}
			\STATE $\tilde{V} \leftarrow \{v \:|\: \mbox{ $v$ is free or $(v,v')\in M$ and $y(v')=0$} \}$.
			\STATE Mark vertices in  $V \setminus \VODD(\tilde{V}, G[1,3])$ as {\em adjustable}.
			\STATE $X_L \leftarrow \{u \:|\: \mbox{ $(u,v) \in B'$ and $u$ is a left adjustable vertex}\}$, 
\\$X_R \leftarrow \{u \:|\: \mbox{ $(u,v) \in B'$ and $u$ is a right adjustable vertex}\}$.
			\STATE If $|X_R| \geq |X_L|$, then let $X \leftarrow X_R$; otherwise let $X \leftarrow X_L$.
			\STATE {}
			\STATE {\bf Perform dual adjustments, starting at $X$:}
			\STATE {$y(v) \leftarrow \begin{cases}
					y(v) - \delta_i & \mbox{if $v \in \VEVEN(X, G[1,3])$}\\
					y(v) + \delta_i & \mbox{if $v \in \VODD(X, G[1,3])$}\\
					y(v) & \mbox{otherwise}
					\end{cases}$}
		\end{algorithmic}
	\end{algorithm}

\begin{figure}[h]
\centering
\begin{tabular}{cc}
\scalebox{.45}{\includegraphics{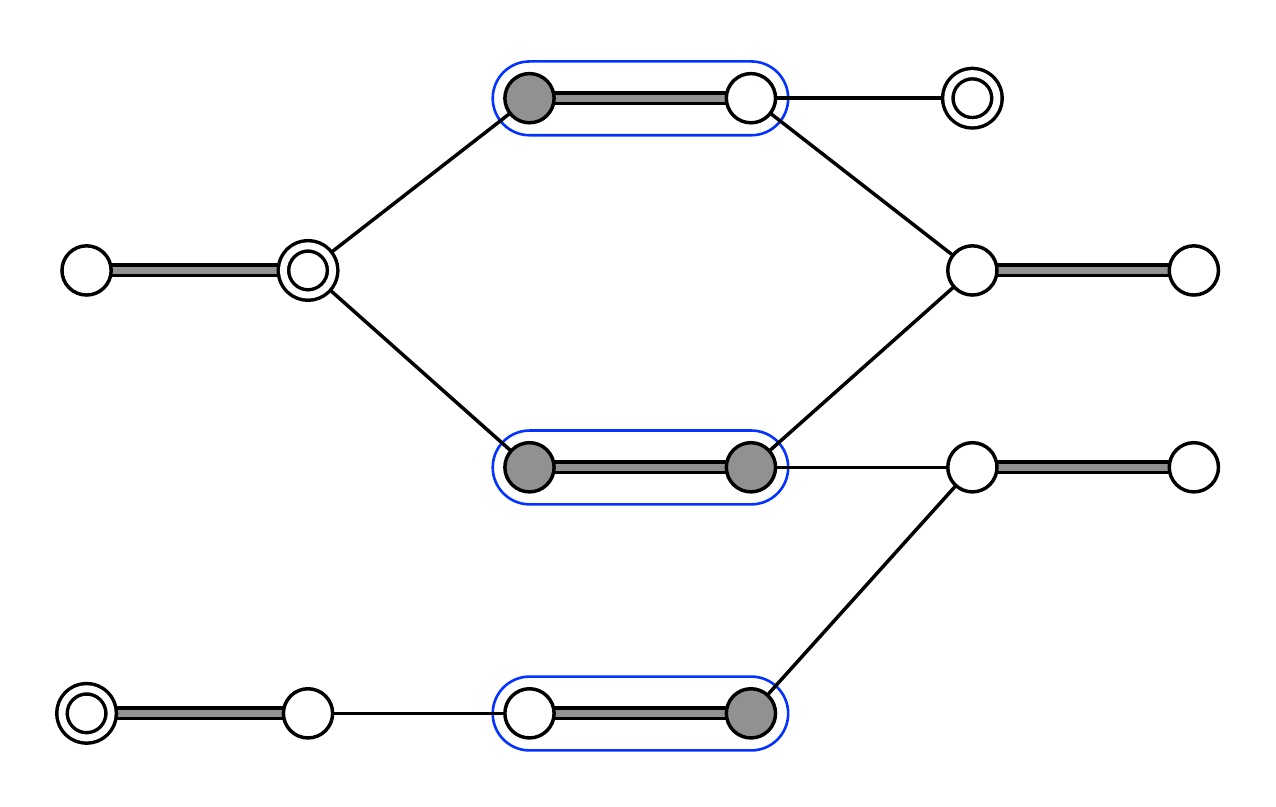}}
& \scalebox{.45}{\includegraphics{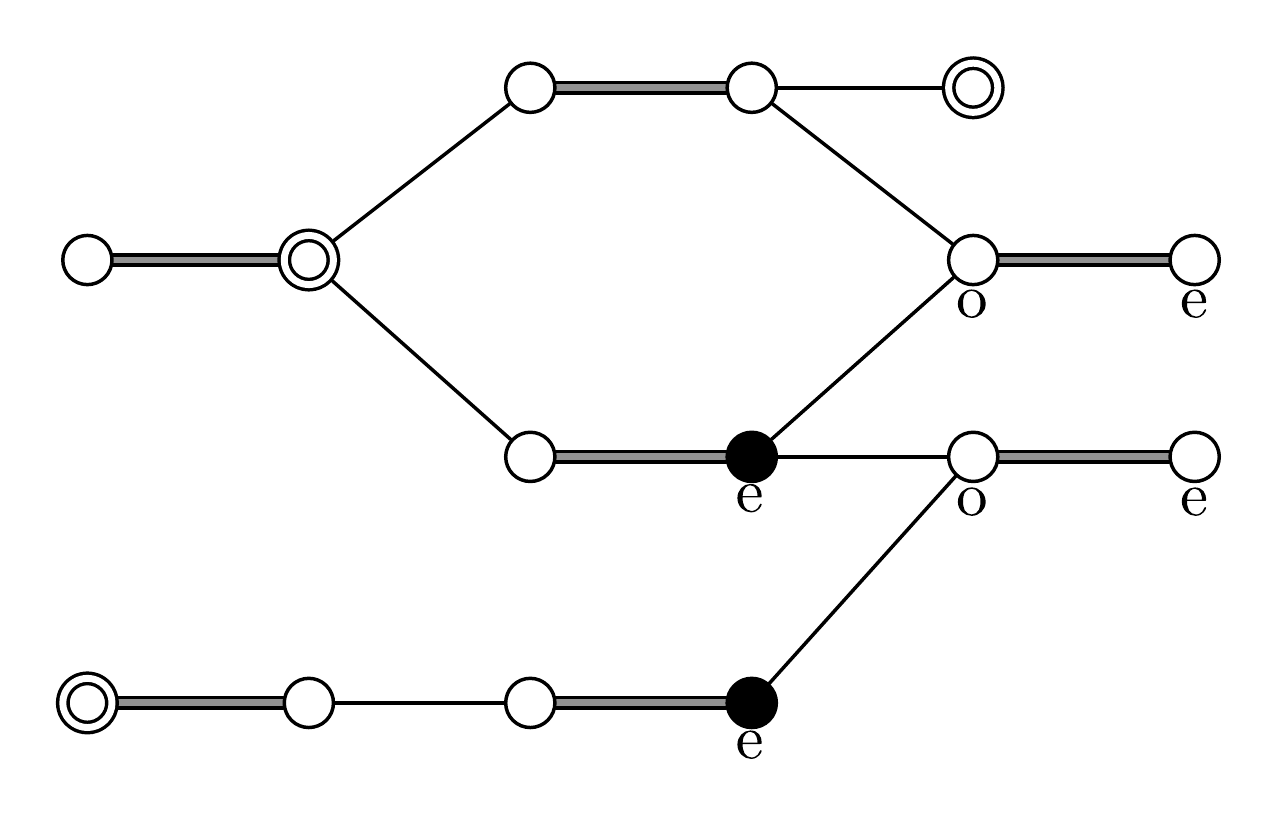}}\\
(a) & (b)
\end{tabular}
\caption{\label{fig:anti_chain} 
(a) The haloed vertices have zero $y$-values. The circled matched edges form an antichain $B'$. 
The shaded vertices in $V(B')$ are dual adjustable.
(b) The black vertices are in $X$ and
vertices marked with `e' and `o' are in $\VEVEN(X,G[1,3])$ and $\VODD(X,G[1,3])$, respectively.}
\end{figure}

\begin{lemma}\label{dual_adjust_anti} The dual adjustment starting at $X$ will not break Property \ref{invariant}(\ref{granularity},\ref{domination},\ref{freeduals}).
Furthermore, it makes Property \ref{invariant}(\ref{tightness}) tighter by decreasing $f(M)$ by $|X|$. \end{lemma}
\begin{proof}
Since $X$ consists of adjustable vertices, 
every vertex $v \in \VEVEN(X, G[1,3])$ must have $y(v) > 0$, implying $y(v)$ will be non-negative after being decremented by $\delta_i$. 
Thus, Property \ref{invariant}(\ref{granularity}) is maintained.  Furthermore, $\VODD(X, G[1,3])$ cannot contain a free vertex, which implies that Property \ref{invariant}(\ref{freeduals}) is preserved. Since all vertices in $X$ are on the same left/right side, $y(e)$ can change by at most $\delta_i$.  
Property~\ref{invariant}(\ref{domination}) (domination) is maintained since no tight edge (with $y(e)=w_i(e)$) can
can have one endpoint in $\VEVEN(X,G[1,3])$ without the other being in $\VODD(X,G[1,3])$.
The algorithm does not increase $f(M)$ since, for any $e\in M\cap G[1,3]$, if one endpoint of $e$ appears in $\VODD(X,G[1,3])$, the other must appear in $\VEVEN(X,G[1,3])$.
Furthermore, the algorithm decreases $f(M)$ by at least $|X|\ge |B'|/2$ since $f(u,v)$ is decremented for each $(u,v)\in B'$ with $v\in X$.
To see this, note that $y(v)$ is decremented by $\delta_i$ and, since $\vec{G}[1,3]$ is acyclic and $B'$ is an antichain, 
$u$ cannot appear in $\VODD(X,G[1,3])$ and therefore cannot have its $y$-value adjusted.
\end{proof}

Therefore, by doing the dual adjustment starting at $X$, we can decrease $f(M)$ by at least $|B'| / 2$.

\subsubsection {Phase II Dual Adjustment: Chain Case}
In the chain case we are given a chain $B'\subseteq B$ and a minimal alternating path $P$ containing $B'$, that is, it starts and ends with $B'$-edges.
Setting $M\leftarrow M\symdiff P$ immediately reduces $f(M)$ by $f(B')$ since $B'$-edges are replaced by tight edges, which contribute nothing to $f(M)$.
However, the endpoints of $P$, say $u$ and $v$, are now free while possibly having positive $y$-values, which violates Property~\ref{invariant}(\ref{freeduals}). 
Our goal is to restore Property~\ref{invariant}, either by finding augmenting paths that rematch $u$ and $v$ or by reducing their $y$-values to zero.  
In this section an {\em augmenting path} is one in the eligibility graph $G[0,3]$ (not $G[1,3]$) 
that either connects $u$ and $v$ or connects an $x\in\{u,v\}$ to a vertex with zero $y$-value.
A notable degenerate case is when $y(x)=0$, in which case the empty path is an augmenting path from $x$ to $x$.
We begin by performing dual adjustments (as in a Hungarian search) until an augmenting path $P_u$ in $G[0,3]$ containing $u$ emerges.
We do not augment along $P_u$ immediately but perform a second search for an augmenting path $P_v$ containing $v$.  
If $P_u$ and $P_v$ do not intersect we let $Q = P_u \cup P_v$.  On the other hand, 
if $P_u$ and $P_v$ do intersect then there must be an augmenting path $P_{uv}$ between $u$ and $v$ in $G[0,3]$;
we let $Q=P_{uv}$.  In either case we augment along $Q$, setting
$M\leftarrow M\symdiff Q$.  Figure \ref{fig:chain} illustrates the case where $Q=P_{uv}$.

The search from $x \in \{u, v\}$ works as follows.  
If there exists an augmenting path $P_x$ in $G[0,3]$ starting at $x$, then return $P_x$.  (This will be the empty path if $y(x)=0$.)
Otherwise update $y$-values as follows:
\[
y(z) \leftarrow \begin{cases}
					y(z) - \delta_i & \mbox{if $z \in \VEVEN(x, G[0,3])$}\\
					y(z) + \delta_i & \mbox{if $z \in \VODD(x, G[0,3])$}\\
					y(z) & \mbox{otherwise}
					\end{cases}
\]
and continue to perform these dual adjustments until an augmenting path starting from $x$ emerges.
As in standard Hungarian search, this process is reducible to single source shortest paths on a non-negatively weighted directed graph.
The graph is either $\vec{G}$ or its transpose, depending on the left/right side of $x$.  The weight function $\hat{w}$ is zero on $M$-edges and equal to the gap
between $y$-value and weight on non-$M$ edges: $\hat{w}(e) = y(e) - w_i(e)$.  By Property~\ref{invariant}(\ref{domination}) (domination), $\hat{w}$ is non-negative.
The value $h(z)$ is the sum of dual adjustments until an augmenting path emerges
from $x$ to $z$.  If $z$ is free and on the opposite side of $x$ then this is simply the distance from $x$ to $z$.
However, if $z$ is on the same side as $x$ (and therefore matched), we need another $y(v)/\delta_i$ dual adjustment steps
to reduce its $y$-value to zero.  
The pseudocode for {\em Search}$(x)$ appears below.

	\begin{algorithm}
		\caption{{\em Search}$(x)$}
		\label{alg5}
		\begin{algorithmic}
			\STATE {\bf Initialize the weighted graph:}
			\STATE {$\hat{G} \leftarrow \begin{cases}
				\vec{G} 		& \mbox{ if $x$ is a left vertex}\\
				\vec{G}^T	& \mbox{ if $x$ is a right vertex (reverse the orientation)}
			\end{cases}$}
			\STATE {$\hat{w}(e) \leftarrow \begin{cases} y(e) - w_i(e) & \text{if $e \notin M$} \\ 0 & \text{if $e \in M$} \end{cases}$}
			\STATE {\bf Compute distances.  For each $z\in V$:}
			\STATE {$d(z) \leftarrow $ distance from $x$ to $z$ in $\hat{G}$ with respect to $\hat{w}$, or $\infty$ if $z$ is unreachable from $x$.}
			\STATE {$h(z) \leftarrow \begin{cases} d(z) & \text{if $z$ is free and not on the same side as $x$} \\ 
				d(z) + y(z) & \text{if $z$ is on the same side as $x$}  \\
				\infty & \text{otherwise}
			\end{cases}$}
			\STATE {$z_{\min} \leftarrow \argmin_{z\in V} \{h(z)\}$}
			\STATE {$\Delta \leftarrow h(z_{\min})$}
			\STATE {\bf Perform dual adjustments.  For each $z\in V$:}
			\STATE $y(z) \leftarrow \begin{cases}y(z) - max\{0, \Delta - d(z)\} & \text{if $z$ is on the same side as $x$}\\ y(z) + max\{0, \Delta - d(z)\} & \text{if $z$ is not on the same side as $x$} \end{cases}$
			\STATE $P_x \leftarrow$ a shortest path from $x$ to $z_{\min}$.  (Note: $\hat{w}(P_x)=0$ after dual adjustments above.)
			\RETURN $P_x$

		\end{algorithmic}
	\end{algorithm}

\begin{figure}[h]
\centering
\begin{tabular}{c@{\hcm[1]}c}
\scalebox{.45}{\includegraphics{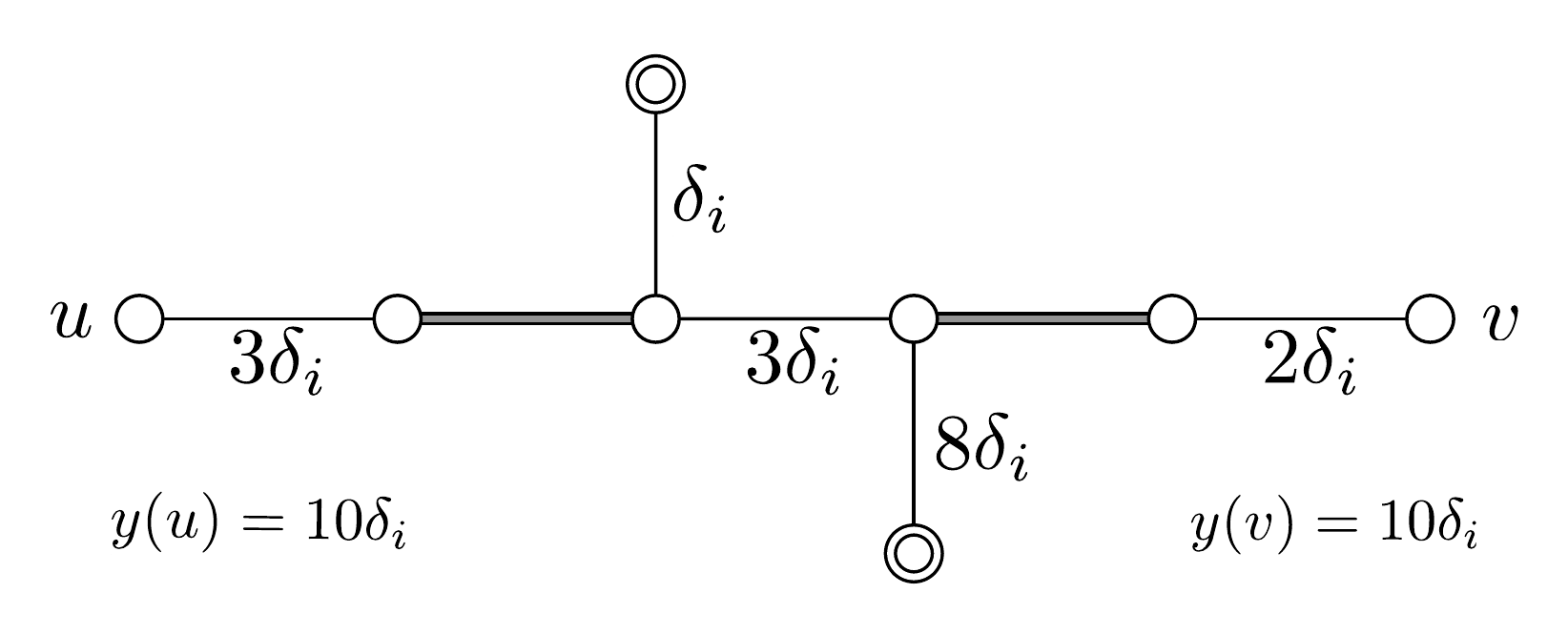}}
& \scalebox{.45}{\includegraphics{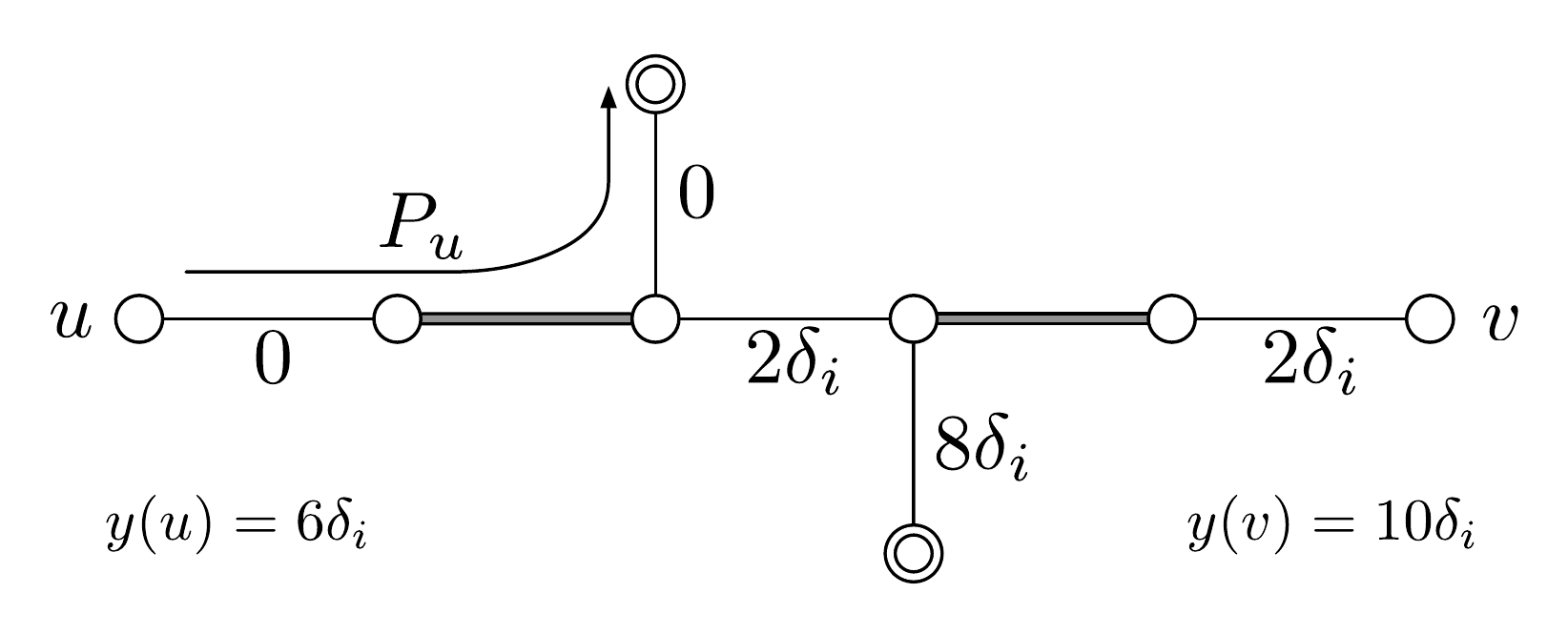}}\\
(a) & (b)\\&\\
\scalebox{.45}{\includegraphics{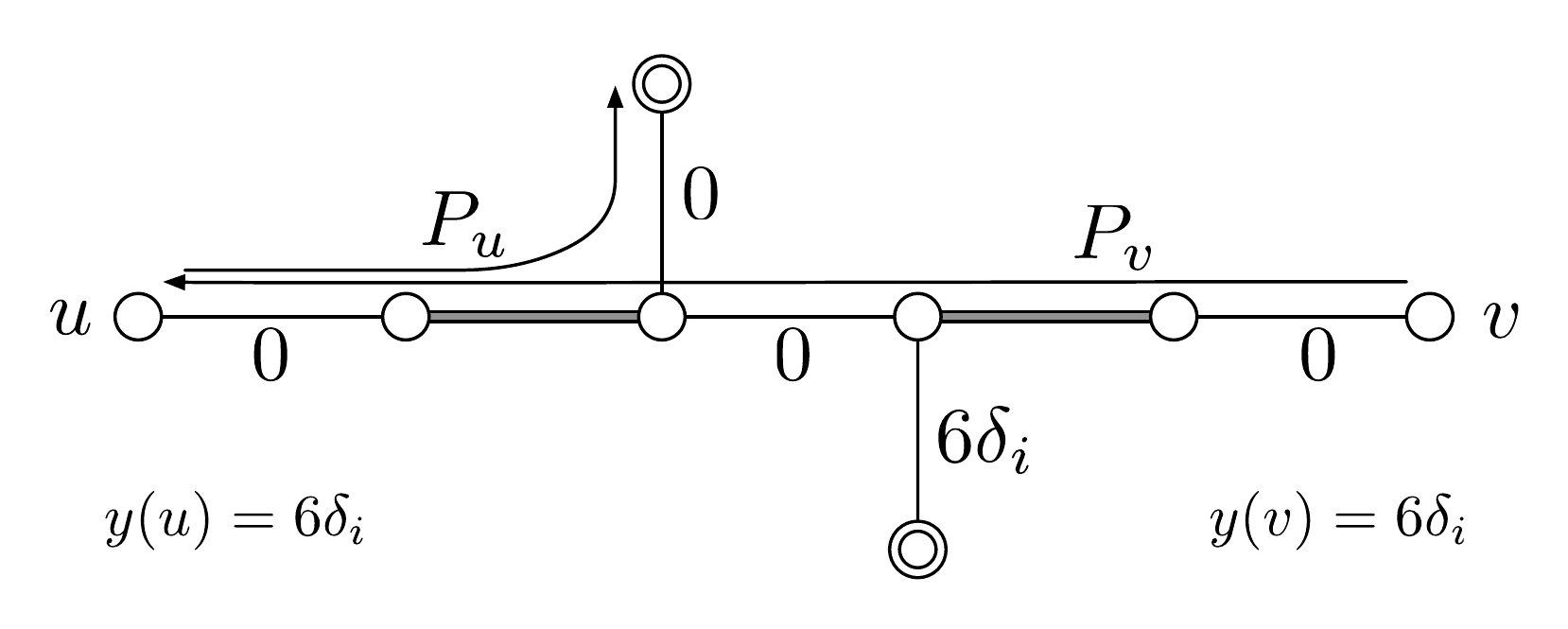}}
&\scalebox{.45}{\includegraphics{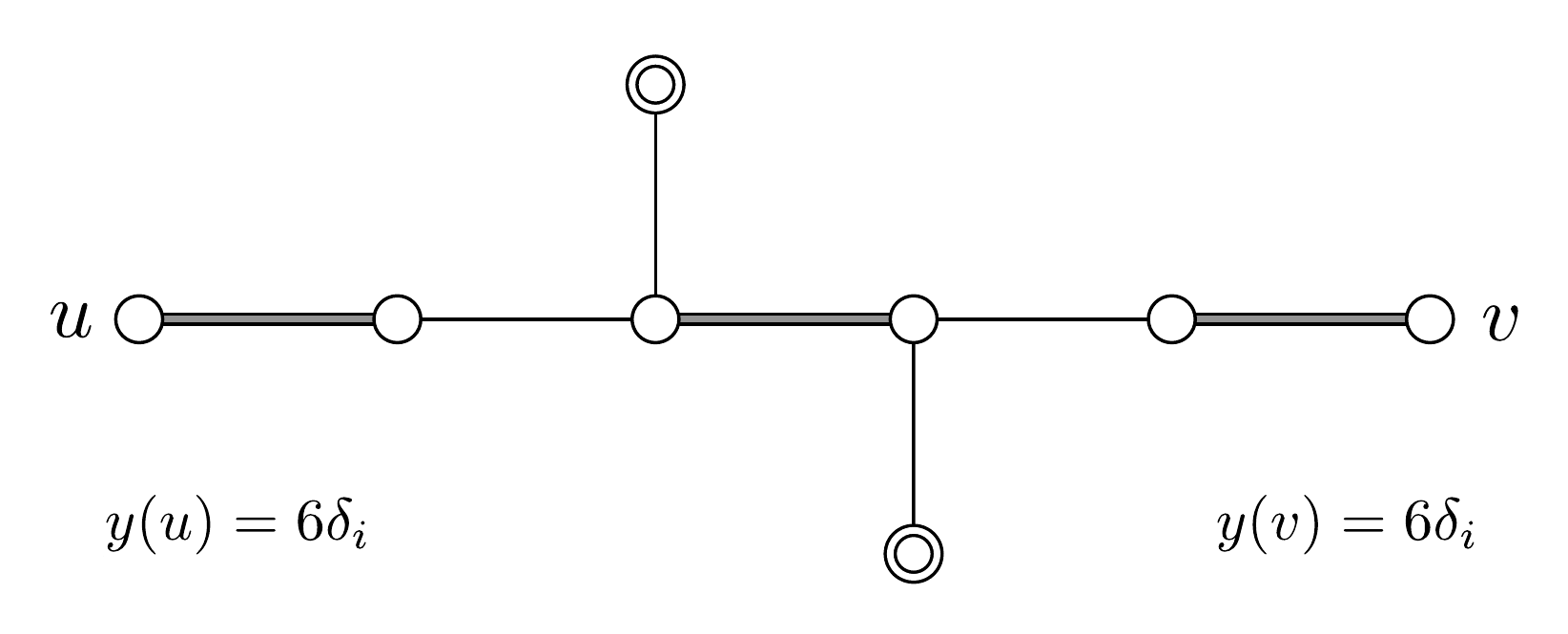}}\\
(c) & (d)
\end{tabular}
\caption{\label{fig:chain} An example illustrating procedures for the chain case. Edges are shown with their new weight $\hat{w}$. The haloed vertices are free vertices with zero $y$-values. (a) After augmenting along $P$, $u$ and $v$ become free while having positive $y$-values. (b) {\em search}$(u)$ adjusts duals by $\Delta_u = 4\delta_i$ and 
finds an augmenting path $P_u$. (c) {\em search}$(v)$ adjusts duals by $\Delta_v = 4\delta_i$ and finds $P_v$. (d) Augmentation along $Q$. This is the case where there exists an augmenting path $Q$ between $u$ and $v$ in $G[0,3]$, which happens to be $P_v$ in the example.}
\end{figure}

Each execution of {\em Search}$(x)$ clearly takes linear time, except for the computation of the distance function $d$.  We use Dijkstra's algorithm~\citeyear{Dij59}, implementing the priority queue as an array of buckets~\cite{Dial69}.  
Since $\hat{w}(e)/\delta_i$ is an integer and we are only interested in distances at most $\Delta$ (see the pseudocode of {\em Search}), a $(\Delta/\delta_i)$-length array suffices.
Lemma~\ref{chain} implies that $\Delta=O(n\delta_i)$, which gives a total running time of $O(m+n)$.

\begin{lemma}\label{chain}
Augmenting along $P$ then $Q$ does not decrease the weight of the matching, that is, $w((M\symdiff P)\symdiff Q) \ge w(M)$.
Furthermore, $\Delta_u + \Delta_v \le 3n\delta_i$, where $\Delta_x$ is the sum of dual adjustments performed by
{\em Search}$(x)$.
\end{lemma}

\begin{proof}
Call $M_1=M\symdiff P$ and $M_2 = M_1 \symdiff Q$ the matchings after augmenting along $P$ and then $Q$
and let $\tilde{w}$ be the weight function 
$\tilde{w}(e) = y(e) - w_i (e)$. (Notice that $\tilde{w}$ differs from $\hat{w}$ on the matched edges.) 
For a quantity $q$ denote its value before {\em Search}$(u)$ and {\em Search}$(v)$
by $q_{old}$ and after both searches by $q_{new}$. After the two searches, we must have:
\begin{align} 
w_i( Q \setminus M_1)
			&=  \sum_{e \in Q \setminus M_1} y_{new}(e) & \mbox{tightness on unmatched edges} \nonumber \\
			&=  y_{new}(u) + y_{new}(v) + \sum_{e \in M_1 \cap Q} y_{new}(e)   \label{ln:chain} \\
			&=  \zero{y_{new}(u) +y_{new}(v) + w_i(M_1 \cap Q) + \tilde{w}_{new}(M_1\cap Q)} & \mbox{defn. of $\tilde{w}_{new}$} \nonumber
\end{align}

Line (\ref{ln:chain}) follows since, aside from $u$ and $v$, $V(Q\setminus M_1)$ and $V(M_1\cap Q)$ differ only on vertices
with zero $y$-values.  (These are the other endpoints of $P_u$ and $P_v$ when $Q=P_u\cup P_v$.)
Therefore, 

\begin{align}
w_i(M_2) &= w_i(M_1) + w_i(Q\setminus M_1) - w_i(M_1\cap Q)\nonumber\\
 		&= w_i(M_1) + y_{new}(u) +y_{new}(v) + \tilde{w}_{new}(M_1\cap Q)\label{eq:chain} 
\end{align}

A similar proof shows that {\em before} the two searches, we have

\begin{equation}\label{eq:chain2} 
w_i(M) = w_i(M_1) + y_{old}(u) +y_{old}(v) - \tilde{w}_{old}(M\cap P) 
\end{equation}

The total amount of dual adjustment in both searches is 
at most the distance from $u$ to $v$, so $\Delta_{u} + \Delta_{v} \leq \tilde{w}_{old}(M\cap P) \leq 3n \delta_i$.
Moreover:
\begin{align}
w_i(M_2) &\geq w_i(M_1) + y_{new}(u) + y_{new}(v) && \mbox{by (\ref{eq:chain}) and $\tilde{w}_{new}(M_1 \cap Q) \geq 0$}\nonumber \\ 
	&= w_i(M_1) + y_{old}(u) + y_{old}(v) - \Delta_{u} - \Delta_{v}  \nonumber \\
	& \geq w_i(M_1) + y_{old}(u) + y_{old}(v) - \tilde{w}_{old}(M\cap P)  && \nonumber \\
	&= w_i(M) && \mbox{by (\ref{eq:chain2})}\nonumber 
\end{align} 
\end{proof}

\begin{lemma}\label{rounds} At most $O(\sqrt{n})$ rounds of augmentation and dual adjustment are required to reduce $f(M)$ to $0$. \end{lemma}
\begin{proof}
We apply Lemma~\ref{chainanti} with $t = \sqrt{b}/2$, where $b=f(M)$.   It follows that in linear time we can either obtain an anti-chain $B'$ of size at least $\ceil{\sqrt{b}}$ and reduce $f(M)$ by $\ceil{\sqrt{b}/2}$, or obtain a chain $B'$ such that $f(B') \geq \ceil{\sqrt{b}/2}$ and reduce $f(M)$ by $f(B')$. In either case we can reduce $f(M)$ by $\ceil{\sqrt{b}/2}$. 
The number of rounds is at most 
$T(b)$, where $T(0)=0$ and $T(b) = T(b - \ceil{\sqrt{b}/2}) + 1$ for $b > 0$. 
A proof by induction shows that $T(b) \leq 4 \sqrt{b}$, which is at most $4\sqrt{2n}$ since $|M|\le n$ and $f(e)\le 2$ for $e\in M$.
\end{proof}

Phase II concludes at the end of scale $L=\ceil{\log N}$, when $M\cap G[2,3] = \emptyset$.  That is,
when Phase III begins $y(e) = w(e)$ or $w(e) + \delta_L$ for each $e\in M$.  (Note that $w_L = w$ since $\delta_L < 1$.)

\subsection{Phase III}

Like a single scale of Phase II, Phase III alternately executes augmentation and dual adjustment steps.
Certain complications arise since the goal of Phase III is to eliminate all non-tight $M$-edges whereas 
the goal of Phase II was to eliminate edges in $M\cap G[2,3]$.
In order to understand the ramifications of this slight shift, we should review the interplay between 
augmentation and dual adjustment in Phase II.
Note that Phase II augmentation steps improved the weight of $M$, but this was an incidental benefit.
The {\em real} purpose of augmentation was to eliminate any augmentations in $G[1,3]$ (thereby making $\vec{G}[1,3]$ acyclic), which then let us reduce $f(M)$ by $\sqrt{f(M)}/2$ with a chain/antichain dual adjustment.
The efficiency of the augmentation step stemmed from the fact that matched and unmatched edges had different eligibility criteria.
Thus, augmenting along a maximal set of augmentations in $G[1,3]$ destroyed all augmentations in $G[1,3]$, that is, tight $M$-edges could and should be ignored.

In Phase III we cannot afford to exclude tight $M$-edges from the eligibility graph.  This raises two concerns.  First, augmenting along a maximal set of augmenting paths/cycles in $G[0,1]$ does {\em not} destroy all augmentations in $G[0,1]$.  (Tight $M$-edges remain eligible after augmentation and may therefore be contained in another augmentation.)  Second, $G[0,1]$ may contain cycles of tight edges, that is, augmentations that do not improve the weight of $M$, so eliminating all {\em weight}-augmenting paths/cycles does not guarantee that $\vec{G}[0,1]$ is acyclic.
These concerns motivate us to redefine {\em augmentation}.  In Phase III an eligible augmentation is either an alternating cycle or alternating path whose endpoints have zero $y$-values, that, in addition, 
is contained in $G[0,1]$ but not $G[0,0]$.  That is, it cannot consist solely of tight edges.

\subsubsection{Phase III Augmentation}

In a Phase III Augmentation step we repeatedly augment along eligible augmentations in $G[0,1]$
until no such eligible augmentation exists.  
Lemma~\ref{lasttrim} lets us upper bound the aggregate time for all Phase III Augmentation steps.

\begin{lemma}\label{lasttrim} 
An eligible augmentation $P$ in $G[0,1]$ can be found in $O(m)$ time, if one exists,
and $w(M\symdiff P) > w(M)$.
Consequently, there can be at most $\sqrt{n}$ augmentations in Phase III. 
\end{lemma}
\begin{proof}
To find an augmentation we begin by computing the strongly connected components of $\vec{G}[0,1]$ in linear time.
If the endpoints of an $e\in M\cap G[1,1]$ are in the same strongly connected component then that edge is contained in an augmenting cycle.
If there are no augmenting cycles, use the linear time algorithm described in Lemma \ref{uvadjust} to determine whether $v$ is dual adjustable for all $v \in V$. If {\em both} endpoints of some $e \in M\cap G[1,1]$ are {\em not} dual adjustable,
then there must be an augmenting path containing $e$.  

Let $P$ be an augmenting path or cycle.
We must have $\sum_{v \in V(P \cap M)} y(v) = \sum_{v \in V(P\setminus M)} y(v)$, since the sums differ only
on vertices with zero $y$-values.
Thus 
$w(P \cap M) < \sum_{v \in V(P \cap M)} y(v) = \sum_{v \in V(P\setminus M)} y(v) = w(P\setminus M)$,
where the first inequality follows from the fact that $P\cap M$ contains a non-tight edge
and the last equality from the tightness of unmatched edges.
Since all weights are integers, the weight of the matching is increased by at least one.

Recall that $\delta_{L} = 2^{\floor{\log N/\sqrt{n}} - \ceil{\log N}} \leq 1 / \sqrt{n} $. By Lemma \ref{firstlemma}, we have 
$w(M) \geq w(M^*) - n \delta_{L} \geq w(M^*) - \sqrt{n}$ at the end of Phase II, where $M^*$ is the \MWM. 
As in Phase II, Phase III dual adjustments will not reduce the weight of $M$.  (See Lemma~\ref{chain}.)
Since each augmentation increases the weight of $M$ there can be at most $\sqrt{n}$ augmentations in Phase III.
\end{proof}

Lemma~\ref{lasttrim} implies that the time for $\ell$ Phase III Augmentation steps is $O(\ell m + m\sqrt{n})$: linear time per step plus linear time per augmentation discovered.

\subsubsection{Phase III Dual Adjustment}

Since the goal of Phase III is to eliminate edges in $M\cap G[1,1]$ (as opposed to $M\cap G[2,3]$) 
we must redefine the {\em badness} of an edge accordingly.  Let $B= M\cap G[1,1]$ be the bad edges
and let $f : E\rightarrow \{0,1\}$ measure badness, where $f(e) = (y(e)-w(e))/\delta_L$ if $e\in B$ and zero if $e\not\in B$.
We define $B'\subseteq B$ to be a chain or antichain exactly as in Phase II.  The difference is that $\vec{G}[0,1]$ is not necessarily acyclic so {\em finding} a chain or antichain requires one extra step.  

\begin{lemma}\label{chainanti-phaseIII}
For any $t>1$, there exists a $B'\subseteq B$ such that $B'$
is a chain with $f(B') \ge \ceil{t}$ or $B'$ is an antichain with $|B'| = f(B') \ge \ceil{f(M)/t}$. 
Moreover, $B'$ can be found in linear time.
\end{lemma}

\begin{proof}
Let $\vec{G}^*$ be the graph obtained from $\vec{G}[0,1]$ by contracting all strongly connected components in $\vec{G}[0,0]$.  Since $G[0,1]$ contains no eligible augmenting cycles after a Phase III Augmentation step, all $B$-edges straddle different strongly connected components and therefore remain in $\vec{G}^*$.  By definition $\vec{G}^*$ is acyclic.
The argument from Lemma~\ref{chainanti} shows that $\vec{G}^*$ contains a chain $B'$ with $f(B')\ge \ceil{t}$
or an antichain with $|B'| = \ceil{f(M)/t}$, and that such a $B'$ can be found in linear time.\footnote{The bound on $|B'|$ is $\ceil{f(M)/t}$ rather than $\ceil{f(M)/2t}$ since
the range of $f$ is $\{0,1\}$ rather than $\{0,1,2\}$.}
\end{proof}

We can apply the chain and antichain elimination procedures from Phase II without compromising correctness
since 
Lemmas~\ref{uvadjust}, \ref{dual_adjust_anti}, \ref{chain}, and \ref{rounds} remain valid if we substitute $G[0,1]$ for $G[1,3]$.
Let $t=\sqrt{b/2}$, where $b = f(M)$ is the current total badness.  Lemmas~\ref{uvadjust} and \ref{chainanti-phaseIII} imply that we can reduce
$f(M)$ by $\ceil{t} \ge \sqrt{b/2}$ in the chain case or $\ceil{b/2t} \ge \sqrt{b/2}$ in the antichain case.  The number
of augmentation and dual adjustment steps in Phase III is then $T(b)$ where $T(b) = T(b - \ceil{\sqrt{b/2}}) + 1$ and $T(0)=0$.
By induction $T(b) \le 2\sqrt{2b}$, which is at most $2\sqrt{2n}$ since $f(M)\le n$.
Thus, the total time spent on dual adjustment in Phase III is still $O(m\sqrt{n})$.

\subsection{Maximum Weight Perfect Matching}\label{sec:MWPM}

Recall from Section~\ref{intro} that the maximum weight perfect matching problem (\MWPM) is reducible to \MWM.
One simply adds $nN$ to the weight of every edge; a \MWM{} in the new graph is necessarily a \MWPM{} in the original.
Thus, our \MWM{} algorithm solves the \MWPM{} problem in $O(m\sqrt{n}\log(nN))$ time, where the number of scales is $\log((n+1)N)$.  However, we can circumvent this roundabout reduction and solve \MWPM{} directly, in $\log(\sqrt{n}N)$ scales, that is, a factor 2 improvement for small $N\ll n$.  We substitute Property~\ref{invariant-mwpm} for 
Property~\ref{invariant}.

\begin{property}\label{invariant-mwpm} 
Redefine $\delta_{0} = 2^{\floor{\log N}}$, $L = \ceil{\log (\sqrt{n}N)}$.
In each scale $i \in [0,L]$, we maintain a {\em perfect} matching $M$ satisfying the following properties.
\begin{enumerate}
\item \emph{Granularity:} $y(u)$ is a multiple of $\delta_i$ for all $u\in V$.
\item \emph{Domination:} $y(e) \geq w_{i}(e)$ for all $e \in E$.
\item \emph{Near Tightness:} For any $e\in M$, $y(e) \leq w_{i}(e) + 3 \delta_{i}$ throughout scale $i$
and $y(e) \leq w_{i}(e) + \delta_{i}$ at the end of scale $i$.
\end{enumerate}
\end{property}

In Phase I we find any perfect matching $M$ in $O(m\sqrt{n})$ time~\cite{HK73,Dinic70,Karzanov73a}
and initialize $y$-values to satisfy Property~\ref{invariant-mwpm}.
\[
y(u) \leftarrow \begin{cases}\delta_0 & \mbox{ if $u$ is a left vertex}\\
0 & \mbox{ if $u$ is a right vertex}\end{cases}
\]
Since $y(e) \le w_i(e) + \delta_0$ for $e\in M$, Property~\ref{invariant-mwpm} lets us end scale 0.

As before, Phase II operates at scales $i \in [1,L]$ and Phase III operates at scale $L$ with the following simplifications:
\begin{enumerate}
\item To begin scale $i$ we simply 
increment $y(u)$ by $\delta_i$ for each left vertex $u$.  There is no need for an initial augmentation/dual adjustment
(as in Section~\ref{sect:phaseII}) since there are no free vertices.

\item As there are no free vertices, an {\em augmentation} is always an augmenting cycle.  In the Phase II and Phase III
augmentation steps we use only {\em Cycle-Search}, not {\em Path-Search}.

\item We do not prohibit negative $y$-values.  Thus, in the antichain case of dual adjustment, both endpoints of
a $B'$-edge are dual-adjustable and we can reduce $f(M)$ by $|B'|$ rather than just $|B'|/2$.

\item In the chain case of dual adjustment, setting $M\leftarrow M\symdiff P$ temporarily frees $P$'s endpoints, say $u$ and $v$.
We execute {\em Search}$(u)$ but force $z_{\min} = v$, that is {\em Search} returns a $u$-to-$v$ path $P_u$.  Setting 
$M\leftarrow M \symdiff P_u$ restores the perfection of $M$.  There is no need for {\em Search}$(u)$ to calculate $h$-values.  These were introduced to maintain Property~\ref{invariant}(\ref{granularity}), namely that $y$-values are non-negative.

\item We change the parameter that determines whether we apply a chain or antichain dual adjustment.
Choose $t = \sqrt{b/2}$, where $b=f(M)$.  Either we can obtain an anti-chain $B'$ of size at least $\ceil{\sqrt{b/2}}$ 
and decrease $f(M)$ by $\ceil{\sqrt{b/2}}$, or we can obtain a chain $B'$ such that $f(B') \geq \ceil{\sqrt{b/2}}$ 
and decrease $f(B)$ by $f(B')$. Thus, the number of rounds is at most 
$T(b) = T(b - \lceil \sqrt{b}/2 \rceil) + 1 \le 2\sqrt{2b}$.  In a Phase II scale $b\le 2n$, so $4\sqrt{n}$ augmentation/dual adjustment steps are needed.  In Phase III $b\le n$ so $2\sqrt{2n}$ steps are needed.

\item Note that since $\delta_L = 2^{\floor{\log N} - \ceil{\log (\sqrt{n}N)}} \leq 1 / \sqrt{n}$,
at the end of Phase II we have $w(M) \geq w(M^{*}) - n \delta_L \geq w(M^*) - \sqrt{n}$, 
by Lemma~\ref{firstlemma}. Thus, the matching can be augmented at most $\sqrt{n}$ times in Phase III
and the total time spent in augmentation steps is $O(m\sqrt{n})$.
\end{enumerate}

The main difference between our \MWM{} and \MWPM{} algorithms is in Phase I.  In the \MWM{} algorithm we 
can afford to use a smaller value for $\delta_0$ since Phase I ends when free vertices have zero $y$-values, whereas 
Phase I of the \MWPM{} algorithm ends only when we have a perfect matching.

\bibliographystyle{plain}

\end{document}